\documentclass[acmsmall,screen]{acmart}  %

\newif\iffullversion
\fullversiontrue

\iffullversion
\else 
\fi

\acmPrice{}
\acmDOI{10.1145/3428220}
\acmYear{2020}
\copyrightyear{2020}
\acmSubmissionID{oopsla20main-p95-p}
\acmJournal{PACMPL}
\acmVolume{4}
\acmNumber{OOPSLA}
\acmArticle{152}
\acmMonth{11}

\setcopyright{rightsretained}

\usepackage[utf8]{inputenc}
\usepackage[T1]{fontenc}
\usepackage{microtype}

\usepackage{booktabs}   
\usepackage{subcaption} 
\usepackage{wrapfig}  

\usepackage{graphicx}
\usepackage{proof}
\usepackage{extarrows}
\usepackage{listings}

\usepackage{datetime}
\usepackage{dsfont}
\usepackage{stmaryrd}     

\usepackage[version=0.96]{pgf}
\usepackage{tikz}
\usetikzlibrary{calc}
\usetikzlibrary{automata,positioning,decorations.markings}
\usetikzlibrary{arrows}
\usepgflibrary{arrows}
\usetikzlibrary{arrows.meta}

\usepackage{verbatim}

\newcommand{\removebeforesubmission}[1]{#1} 
\renewcommand{\removebeforesubmission}[1]{} 

\newcommand{\fullversionref}[1]{%
\iffullversion%
Appendix~\ref{#1}%
\else
\cite{Igloo-TR-2020}%
\fi%
}

\definecolor{persianblue}{rgb}{0.11, 0.22, 0.73}
\definecolor{navyblue}{rgb}{0.0, 0.0, 0.5}
\definecolor{darkpowderblue}{rgb}{0.0, 0.2, 0.6}
\definecolor{frenchblue}{rgb}{0.0, 0.45, 0.73}
\definecolor{burntorange}{rgb}{0.8, 0.33, 0.0}

\newcommand{\fix}[2]{\footnote{\textbf{#1}: #2}}
\newcommand{\marking}[2]{\textcolor{#1}{#2}}

\newcommand{\todo}[1]{\marking{red}{[TODO: #1]}}

\renewcommand{\fix}[2]{}  
\renewcommand{\marking}[2]{#2}
\renewcommand{\todo}[1]{}

\newcommand{\tobias}[1]{\fix{Tobias}{#1}}

\newcommand{\christoph}[1]{\marking{blue}{#1}}
\newcommand{\pmue}[1]{\marking{blue}{{#1}}}
\newcommand{\mar}[1]{\marking{blue}{#1}}
\newcommand{\tk}[1]{\marking{blue}{#1}}
\newcommand{\faw}[1]{\marking{blue}{#1}}

\newcommand{\id}[1]{\mathit{#1}}
\newcommand{\kw}[1]{\mathsf{#1}}

\newcommand{\ie}{i.e.} 
\newcommand{\eg}{e.g.}
\newcommand{\cf}{cf.}

\newcommand{\Step}[1]{\textbf{Step~#1}}

\newcommand{\tru}{\kw{true}}
\newcommand{\fal}{\kw{false}}
\newcommand{\Bool}{\mathds{B}}
\newcommand{\Nat}{\mathds{N}}

\newcommand{\option}[1]{#1_{\bot}}      

\newcommand{\pick}{\id{pick}}            

\newcommand{\Unit}{\mathds{1}}       
\newcommand{\unit}{\bullet}

\newcommand{\List}[1]{#1^*}                
\newcommand{\nil}{\epsilon}         
\newcommand{\cons}{\,\#\,}                 
\newcommand{\cc}{\cdot}                     
\newcommand{\snoc}[1]{\cc \mklist{#1}}       
\newcommand{\mklist}[1]{\langle#1\rangle}  
\newcommand{\len}[1]{\mathit{len}(#1)}   

\newcommand{\Pow}{\mathbb{P}}        
\newcommand{\Powerset}[1]{\Pow(#1)}

\DeclareMathOperator*{\bigmultisetsum}{\scalerel*{\displaystyle{\sum}^{\#}}{\sum}} 
\usepackage{scalerel}

\newcommand{\fun}{\rightarrow}
\newcommand{\map}{\rightharpoonup}
\newcommand{\funupd}{\mapsto}

\newcommand{\ift}[2]{\kw{if}\ #1\ \kw{then}\ #2}
\newcommand{\els}{\kw{else}}
\newcommand{\ifte}[3]{\ift{#1}{#2}\ \els\ #3}

\newcommand{\lmparen}{\ensuremath{\mathclose{(\mkern-3.5mu\mid}}} 
\newcommand{\rmparen}{\ensuremath{\mathclose{\mid\mkern-3.5mu)}}} 
\newcommand{\record}[1]{\mbox{$\lmparen\,#1\,\rmparen$}}

\newcommand{\multiset}[1]{#1^\#}
\newcommand{\melem}{\mathrel{\in^\#\!}}
\newcommand{\mempty}{\emptyset^\#}
\newcommand{\mset}[1]{\{#1\}^\#}
\newcommand{\mplus}{\mathrel{+^\#}}

\newcommand{\pc}{\!\downarrow}

\newcommand{\corec}{\nu}

\newcommand{\pointsto}[2]{#1 \mapsto #2}

\newcommand{\Chunk}{\id{Chunks}}
\newcommand{\Place}{\id{Places}}
\newcommand{\Val}{\id{Values}}
\newcommand{\Bio}{\id{Bios}}
\newcommand{\bio}{\id{bio}}
\newcommand{\token}{\id{token}}

\newcommand{\sand}{\star}
\newcommand{\sall}{\forall^{\star}}

\newcommand{\Heap}{\id{Heap}}
\newcommand{\ActBio}{\id{Act}}
\newcommand{\Ty}{\id{Ty}}

\newcommand{\BioRule}{\textsf{Bio}}
\newcommand{\ContradictRule}{\textsf{Contradict}}
\newcommand{\ChaosRule}{\textsf{Chaos}}

\newcommand{\can}{\id{can}} 
\newcommand{\welltyped}{\id{welltyped}}
\newcommand{\addtoken}[1]{#1^{\,\mathrm{t}}}

\newcommand{\Hoare}[3]{\{#1\}\ #2\ \{#3\}}

\newcommand{\trans}[1]{\xrightarrow{#1}}
\newcommand{\traces}{\id{traces}}
\newcommand{\ES}{\mathcal{E}}        
\newcommand{\GES}{\mathcal{G}}        
\newcommand{\ESH}{\mathcal{H}}        
\newcommand{\ESP}{\mathcal{P}}        

\newcommand{\eventhead}[1]{#1:\;\,}
\newcommand{\eventsep}{\triangleright}
\newcommand{\eventbody}[2]{#1 \;\,\eventsep\;\, #2}
\newcommand{\event}[3]{\eventhead{#1}\eventbody{#2}{#3}}

\newcommand{\skipp}{\kw{skip}}         
\newcommand{\eparallel}[1]{\mathrel{\parallel_{#1}}}
\newcommand{\chiI}{\chi_{I}}
\newcommand{\interl}{\eparallel{\chiI}}

\newcommand{\hact}[1]{\xrightarrow{#1}_{\mathcal{H}}}
\newcommand{\tracesH}{\traces^\mathcal{H}}

\newcommand{\refines}[1]{\sqsubseteq_{#1}}

\newcommand{\Proc}{\id{Proc}}
\newcommand{\Null}{\kw{Null}}
\newcommand{\choice}{\oplus}
\newcommand{\bigchoice}{\bigoplus}

\newcommand{\opsem}[1]{\xrightarrow{#1}_{\mathcal{P}}}

\newcommand{\proc}{\id{proc}}       
\newcommand{\emb}{\id{emb}}      

\newcommand{\Left}{\kw{L}}
\newcommand{\Right}{\kw{R}}

\newcommand{\cmodel}{\id{cmod}}
\newcommand{\gmodel}{\id{gmod}}
\newcommand{\premodel}{\id{pm}}

\newcommand{\rhowit}{\rho_{\id{wit}}}
\newcommand{\srcplaces}{\id{srcs}}       

\newcommand{\Nodeid}{\id{ID}}
\newcommand{\nextt}{\id{next}}
\newcommand{\nextnode}[1]{\nextt(#1)}
\newcommand{\winner}{\mathit{leader}}
\newcommand{\chan}{\mathit{chan}}
\newcommand{\ibuf}{\mathit{ibuf}}
\newcommand{\obuf}{\mathit{obuf}}
\newcommand{\Addr}{\mathit{Addr}}
\newcommand{\addrNA}{\id{addr}}
\newcommand{\addr}[1]{\addrNA(#1)}

\newcommand{\elive}{$\textit{live}_{\mathit{env}}$}
\newcommand{\alive}{$\textit{live}$}

\newcommand{\syst}{\mathrm{s}}
\newcommand{\envir}{\mathrm{e}}
\newcommand{\system}[1]{#1^{\syst{}}}
\newcommand{\environ}[1]{#1^{\envir{}}}

\newcommand{\Systmodel}{\system{\ES}}
\newcommand{\Envmodel}{\environ{\ES}}
\newcommand{\Envreal}{\ES^{\mathrm{re}}}

\newcommand{\POSee}{\nil \cc \nil}
\newcommand{\POSel}[1]{\nil \cc #1}
\newcommand{\POSle}[1]{#1 \cc \nil}
\newcommand{\POS}[2]{#1 \cc #2}
\tikzset{
processstyle/.style={
  inner sep=0pt,
  text width=6mm,
  align=center
  }
}
\tikzset{
ppos/.style={
  inner sep=0pt,
  text width=20mm,
  }
}
\tikzset{
    position/.style args={#1:#2 from #3}{
        at=(#3.#1), anchor=#1+180, shift=(#1:#2)
    }
}

\newcommand{\normalDist}{14mm}
\newcommand{\bigDist}{15mm}
\newcommand{\piDist}{9.0cm}

\begin{document}

\title[Igloo: Soundly Linking Compositional Refinement and Separation Logic for Distributed System Verification]
      {Igloo: Soundly Linking Compositional Refinement and Separation Logic for Distributed System Verification}



\author{Christoph Sprenger}
\email{sprenger@inf.ethz.ch}          

\author{Tobias Klenze}
\email{tobias.klenze@inf.ethz.ch}         

\author{Marco Eilers}
\email{marco.eilers@inf.ethz.ch}         

\author{Felix A.~Wolf}
\email{felix.wolf@inf.ethz.ch}         

\author{Peter M\"uller}
\email{peter.mueller@inf.ethz.ch}         

\author{Martin Clochard}
\email{martin.clochard@inf.ethz.ch}         

\author{David Basin}
\email{basin@inf.ethz.ch}         

\affiliation{
  \department{Department of Computer Science}             
  \institution{ETH Zurich}           
  \country{Switzerland}                   
}

\authorsaddresses{Authors' address: Department of Computer Science, ETH Zurich, Switzerland. Email: firstname.lastname@inf.ethz.ch}

\begin{abstract}
Lighthouse projects such as CompCert, seL4, IronFleet, and DeepSpec have demonstrated that full verification of entire systems is feasible by establishing a refinement relation between an abstract system specification and an executable implementation. Existing approaches however impose severe restrictions on either the abstract system specifications due to their limited expressiveness or versatility, or on the executable code due to their reliance on suboptimal code extraction or inexpressive program logics. 

We propose a novel methodology that combines the compositional refinement of abstract, event-based models of distributed systems with the verification of full-fledged program code using expressive separation logics, which support features of realistic programming languages like mutable heap data structures and concurrency. The main technical contribution of our work is a formal framework that soundly relates event-based system models to program specifications in separation logics, such that successful verification establishes a refinement relation between the model and the code. We formalized our framework, \emph{Igloo}, in Isabelle/HOL.

\christoph{Our framework enables the sound combination of tools for protocol development with existing program verifiers.} We report on three case studies, a leader election protocol, a replication protocol, and a security protocol, for which we refine formal requirements into program specifications \christoph{(in Isabelle/HOL)} that we implement in Java and Python and prove correct using the VeriFast and Nagini tools.
\end{abstract}

%

\iffullversion
\else
\fi

\maketitle


\section{Introduction}
\label{sec:intro}

The full verification of entire software systems, formally relating abstract specifications to executable code, is one of the grand challenges of computer science~\cite{Hoare03}. Seminal projects such as seL4~\cite{KleinEHACDEEKNSTW09}, CompCert~\cite{Leroy06}, IronFleet~\cite{DBLP:conf/sosp/HawblitzelHKLPR15}, and DeepSpec~\cite{DBLP:conf/oopsla/Pierce16} have achieved this goal by formally establishing a refinement relation between a system specification and an executable implementation.

Despite this progress, substantial challenges still lay ahead. We posit that techniques for the verification of entire systems should satisfy \pmue{four} major requirements: 

\begin{enumerate}
\item
\textit{End-to-end guarantees}: Verification techniques need to provide system-wide correctness guarantees whose proofs relate global properties ultimately to verified implementations of the system components. 

\item
\textit{Versatility}:  Verification techniques should be applicable to a wide range of systems. In the important domain of distributed systems,
versatility requires (i)~the ability to model different kinds of environments in which the system operates, capturing, for instance, different network properties, fault models, or attacker models, (ii)~support for different flavors of systems, comprising different types of components (such as clients and servers) and allowing an unbounded number of instances per \tk{component type}, and (iii)~\christoph{support for} heterogeneous implementations, for instance, to support the common case that \christoph{clients are  sequential, servers are concurrent, and each of them is implemented in a different language.}

\item
\textit{Expressiveness}: Verification techniques should support \christoph{expressive languages and logics}. In particular, high-level system models and proofs often benefit from the expressiveness of rich formalisms such as higher-order logic, whereas code-level verification needs to target efficiently-executable and maintainable implementations, often in multiple languages. 

\item \christoph{
\textit{Tool interoperability}: \pmue{While it is possible to support the previous three requirements within one generic verification tool, it is advantageous to employ specialized tools, for instance, to obtain a high degree of automation and to leverage existing tools,} \christoph{infrastructure, and expert knowledge.}
This gives rise to the additional requirement of sound interoperability of different verification tools, which is a long-standing challenge in verification.
%
Moreover, integrating tools should ideally not require any modifications to the tools, even though they may support different logics and programming languages.
}
\end{enumerate}
Although existing work has demonstrated that full verification is now feasible, the employed techniques do not meet all of these requirements. 

Some existing approaches~\cite{DBLP:conf/cpp/Koh0LXBHMPZ19} use specifications of individual system components (such as a server), but do not explain how to formally connect them to a global model of the entire system. A global model is necessary to prove system-wide properties, especially in decentralized systems.
%
Others~\cite{DBLP:conf/ifm/OortwijnH19} do not consider the preservation of global model properties down to the implementation. 
Hence, these approaches do not meet our first requirement.

Most existing approaches do not match our versatility requirements. Some target particular types of systems~\cite{DBLP:conf/esop/RahliVVV18,DBLP:conf/popl/LesaniBC16,KleinEHACDEEKNSTW09} or make fixed environment assumptions~\cite{DBLP:conf/cpp/Koh0LXBHMPZ19,DBLP:journals/pacmpl/SergeyWT18}.
%
Moreover, in several works, different component types with unbounded numbers of instances are either not supported or it is unclear whether they are generically supported~\cite{DBLP:conf/sosp/HawblitzelHKLPR15,DBLP:conf/cpp/Koh0LXBHMPZ19}. 
%
Finally, many approaches~\cite{DBLP:conf/sosp/HawblitzelHKLPR15,DBLP:conf/esop/RahliVVV18,DBLP:conf/popl/LesaniBC16,DBLP:journals/pacmpl/SergeyWT18,DBLP:conf/pldi/WilcoxWPTWEA15} prescribe a fixed programming language and, thus, do not support heterogeneous implementations.

Most previous work does not satisfy our expressiveness requirement. Some of them~\cite{HawblitzelHLNPZZ14,DBLP:conf/sosp/HawblitzelHKLPR15}  limit the formalism used for model development to first-order logic, to leverage SMT solvers, which complicates the formalization of common properties such as graph properties. 
Others restrict the executable implementation~\cite{DBLP:conf/esop/RahliVVV18,DBLP:conf/popl/LesaniBC16,DBLP:journals/pacmpl/SergeyWT18,DBLP:conf/pldi/WilcoxWPTWEA15,DBLP:conf/cpp/WoosWATEA16,Leroy06,nfm20-liu} and extract executable code directly from formal models. This guarantees the implementation's correctness, but has several drawbacks. In particular, \christoph{the} extracted code is purely functional or rewriting-based, with sub-optimal performance, \christoph{and} any manual code optimizations invalidate the correctness argument and may compromise the intended behavior. Moreover, code extraction complicates the interaction 
with existing system components and libraries. Other approaches reason about manually-written implementations, but do not employ a modern verification logic~\cite{KleinEHACDEEKNSTW09}, restricting the implementation, for instance, to sequential code, and precluding the use of existing state-of-the-art program verification tools, potentially resulting in low proof automation and non-modular proofs.

\pmue{Finally, most existing approaches require the use of a single tool, typically an interactive theorem prover. }  
\christoph{This may prevent experts in both protocol and program verification from using the highly automated tools they are familiar with and from building on their existing infrastructure.
An exception is~\citet{DBLP:conf/ifm/OortwijnH19}, who combine the Viper verifier~\cite{DBLP:conf/vmcai/0001SS16} with the mCRL2 model checker~\cite{CranenGKSVWW13} to reason about message passing programs.}

\subsubsection*{This Work}

We propose a novel approach that combines the top-down compositional refinement of abstract, event-based system models~\cite{DBLP:journals/tcs/AbadiL91,DBLP:journals/iandc/LynchV95,Abrial10} with the bottom-up verification of full-fledged program code using separation logic~\cite{Reynolds02a}. Our approach satisfies all \christoph{four} of our requirements.
It offers the full expressive power of higher-order logic and the foundational guarantees of interactive theorem provers for developing formal models, as well as the expressiveness and tool support provided by modern program logics.

The core of our approach is a formal framework that \christoph{soundly} relates event-based system models to program specifications in separation logic, such that successful verification establishes a refinement relation between the model and the code. 
The program specifications \christoph{link models and code and at the same time they} decouple models and code, allowing us to support multiple programming languages and verification tools. 
\christoph{This is,} for instance, useful to develop multiple library implementations of a protocol. Moreover, this decoupling
enables a separation of concerns where we can use specialized tools for the separate tasks of model refinement and code verification, tailored to the problem and the programming language at hand.

We focus on the development of \emph{distributed} systems, consisting of an arbitrary number of components (of possibly heterogenous types such as clients and servers) with local states that interact by exchanging messages via an arbitrary, potentially faulty or adversarial environment. Such systems give rise to complex concurrent behaviors. 
In this setting, the program specification of a component's implementation prescribes the component's state changes as well as its I/O behavior \christoph{and is called an \emph{I/O specification}}. 
For this purpose, we employ an existing encoding of I/O specifications into a separation logic
to support assertions that can specify both of these aspects~\cite{DBLP:conf/esop/Penninckx0P15}. This encoding can be used with any logic that offers standard separation logic features, and can thus be used to verify \christoph{components} with mutable heap data structures, concurrency, and other features of realistic programming languages that enable efficient implementations. 

%

\subsubsection*{Approach}

\begin{wrapfigure}{r}{.48\linewidth}
\vspace{-8pt}
\begin{center}
\includegraphics[scale=.27]{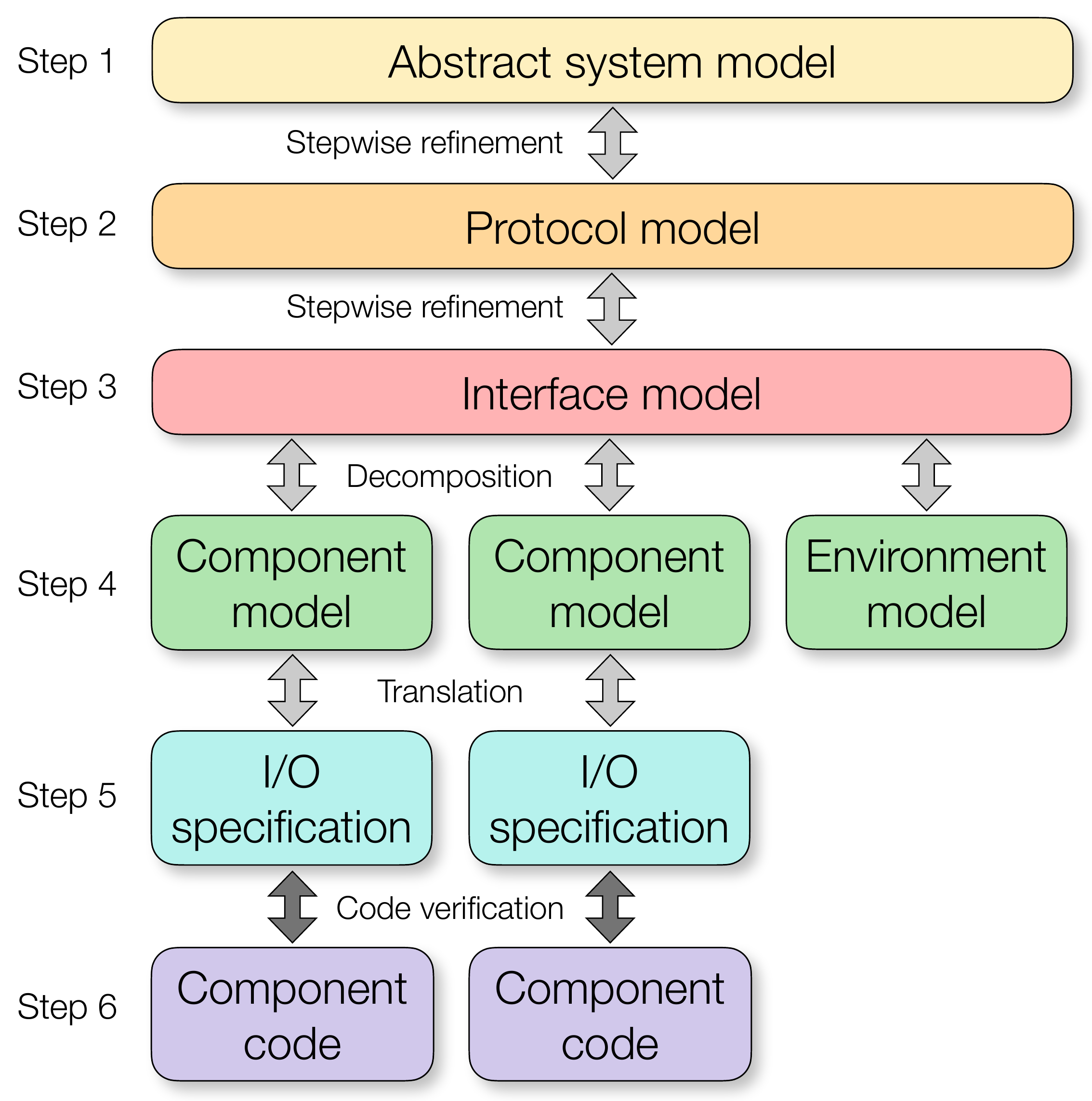}
\caption{The main steps of our approach. Boxes depict formal models, specifications, and programs. Light and dark gray arrows depict proofs in Isabelle/HOL and in program verification tools, respectively.}
\label{fig:approach}
\end{center}
\vspace{-6pt}
\end{wrapfigure}

Our methodology consists of six main steps, illustrated in Figure~\ref{fig:approach}.
All steps come with formal guarantees to soundly link the abstract models with the code. The first five steps are formalized in Isabelle/HOL~\cite{DBLP:books/sp/NipkowPW02}.
\Step{1} requires formalizing an initial abstract model of the entire system and proving desired trace properties. This model and subsequent models are expressed as event systems (\ie, labeled transition systems) in a generic refinement framework that we implemented in Isabelle/HOL\@. 
\Step{2} develops a protocol model, which contains the components of the distributed system to be developed as well as assumptions about the communication network. This \emph{environment} may, for instance, include a fault model or an attacker model, which can be used to prove properties about fault-tolerant or secure systems. 
So far, this is standard development by refinement, but Steps 3-5 are specific to our approach. 
\Step{3} prepares the model for a subsequent decomposition and refines the interfaces of the components and the environment to match the interfaces of the I/O libraries to be used in the implementations. \Step{4} decomposes the, so far monolithic, model into models of the individual system components (\eg, clients and servers) and the environment. 
%
\Step{5} translates each component's event system into an \christoph{I/O specification, which formalizes its valid I/O behaviors.} 
We express this specification as an encoding into standard separation logic assertions that can describe sequences of calls to I/O libraries, \eg, for sending and receiving messages~\cite{DBLP:conf/esop/Penninckx0P15}.  Each such call corresponds to one event of the component's event system. 
Finally, \Step{6} is standard code verification of the \christoph{different} system components, albeit with specifications describing I/O behavior. This verification step can be performed using an embedding of a separation logic into an interactive theorem prover (to obtain foundational guarantees) or by using \christoph{separate} dedicated program verifiers (to increase automation). For the latter, any existing verifier supporting standard separation logic features can be used without requiring changes to the tool, \christoph{provided it satisfies our \emph{verifier assumption}. This assumption states that proving a Hoare triple involving the I/O specification in the tool implies that the program's I/O behavior refines the one defined by the I/O specification.}
\christoph{Crucially, our approach supports modular reasoning in that the verification of a component's code does not involve reasoning about the  system's global properties, other components, or the environment.}
\christoph{Moreover, we can employ different code verifiers to support heterogeneous implementations, where different components are written in different languages, and some are sequential, while others use local concurrency for improved performance.}


Our approach ensures that the resulting distributed system's implementation does not abort due to runtime errors and satisfies, by virtue of compositional refinement, the requirements specified and proved for the formal models. 
%
These guarantees assume that the real environment, including the I/O libraries and the lower software and hardware layers, conforms to our environment model, \christoph{the components are correctly instantiated, and the verification tools used 
are sound.}
%
As our approach ``glues'' together models and code through their I/O behavior, we have dubbed it ``Igloo''. 


\subsubsection*{Contributions}

Our work makes the following contributions:

\begin{enumerate}

\item \textit{Methodology:} We present a novel methodology for the sound end-to-end verification of distributed systems that combines the top-down refinement of expressive, global system specifications with bottom-up program verification based on expressive separation logics. This combination supports the verification of system-wide properties and handles heap data structures, concurrency, and other language features that are required for efficient code. \christoph{Our methodology enables the sound interoperability of interactive theorem provers with \pmue{existing} code verification tools for different programming languages, as well as the verified interoperability of the resulting component implementations.}

\item \textit{Theory:} We establish a novel, formal link between event system models and I/O specifications for programs expressed in separation logics by relating both of them to a process calculus. This link between these disparate formalizations is central to our methodology's soundness. It is also interesting in its own right since it shows how to formally integrate the trace semantics of event systems and processes with the permissions manipulated by separation logics.

\item \textit{Case studies:} We demonstrate the feasibility of our approach by developing formal models for a leader election \christoph{protocol}, a replication \christoph{protocol}, and a security protocol, deriving I/O specifications for their \christoph{components}, and verifying independent implementations in Java and Python, using the  VeriFast~\cite{DBLP:conf/nfm/JacobsSPVPP11} and Nagini~\cite{DBLP:conf/cav/Eilers018} verifiers. \christoph{Some of these components' performance is optimized using locally concurrent execution.}

\item \textit{Formalization:} All our definitions and results are formalized and proven in Isabelle/HOL\@. This includes the refinement framework and its soundness, the formalization of I/O specifications, the soundness proof that formally links event systems, processes, and I/O specifications, and Steps~1--5 of our case studies. This foundational approach yields strong soundness guarantees.
\end{enumerate}

\section{Preliminaries}
\label{sec:preliminaries}


\begin{table}
\caption{Summary of notation.}
\vspace{-10pt}
\label{tab:math-notations}
\begin{small}
\begin{tabular}{c@{\hspace{8mm}}c}
\begin{tabular}[t]{|c|l|}
\hline
$\Unit$, $\Bool$, $\Nat$ & $\{\unit\}, \{\tru,\fal\}$, naturals \\
\hline
$A\times B$ & cartesian product \\
\hline
$\record{x \in A, y\in B}$ & set of records \\
\hline
$A \uplus B$ & disjoint union (sum) \\
\hline
$\option{A}$ & $= A \uplus \{\bot\}$ \\
\hline
$A\fun B$, $A\map B$ & total and partial functions \\
\hline
$\Pow(A)$ & powerset \\
\hline
$A^*$ & finite sequences \\
\hline
$\multiset{A}$ & multisets, $= A \fun \Nat\cup\{\infty\}$ \\
\hline
\end{tabular}
&
\begin{tabular}[t]{|c|l|}
\hline
$\record{x = a, y = b}$ & concrete record \\
\hline
$x(r)$, $r\record{x := v}$ & record field $x$, update \\
\hline
$f(x \funupd v)$, $f^{-1}$ & function update, inverse \\
$\nil$, $abc$ or $\mklist{a,b,c}$ & empty, concrete sequence\\
\hline
$x \cc y$ & concatenation \\
\hline
$\len{x}$, $x(i)$ & length, $i$-th value \\
\hline
$\mset{a,a,b,c}$ & concrete multiset \\
\hline
$M \mplus M'$ & multiset sum \\
\hline
$\mempty$, $a \melem M$ & $\mset{}$, $M(a) > 0$ \\
\hline
\end{tabular}
\end{tabular}
\end{small}
\end{table}

Although we formalize our development in Isabelle/HOL, we use standard mathematical notation where possible to enhance readability. 
Table~\ref{tab:math-notations} summarizes our notation.

\subsection{Event Systems, Refinement, and Parallel Composition}
\label{ssec:event-systems-refinement}

\subsubsection{Event Systems}
\label{ssec:event-systems-and-traces}

An \emph{event system} is a labeled transition system $\ES = (S, E, \trans{})$, where $S$ is a set of states, $E$ is a set of events, and $\,\trans{}\;  \subseteq S \times E \times S$ is the transition relation. We also write $s \trans{e} s'$ for $(s,e,s') \in\;\trans{}$. 
We extend the transition relations to \christoph{finite sequences of events $\tau$} by inductively defining, for all $s, s', s'' \in S$, $s \trans{\nil} s$ and $s\trans{\tau \snoc{e}} s''$, whenever $s \trans{\tau} s'$ and $s' \trans{e} s''$. Given a set of initial states $I \subseteq S$, a \emph{trace} of an event system $\ES$ starting in $I$ is a finite sequence $\tau$ such that $s \trans{\tau} s'$ for some initial state $s\in I$ and reachable state $s'$. We denote by $\traces(\ES,I)$ the set of all traces of $\ES$ starting in $I$. For singleton sets $I=\{s\}$, we also write $\traces(\ES,s)$, omitting brackets. We call a set of traces $P$ over $E$ a \emph{trace property} and write $\ES,I \models P$ if $\traces(\ES,I) \subseteq P$.

For concrete specifications, we often use \emph{guarded event systems} of the form $\GES = (S, E, G, U)$ where $G$ and $U$ denote the $E$-indexed families of \emph{guards} $G_e: S \fun \Bool$ and \emph{update} functions $U_e: S \fun S$. The associated transition relation is $\,\trans{} \; = \{(s,e,s') \mid G_e(s) \land s' = U_e(s)\}$.
%
If $S = \record{\bar{x} \in \bar{T}}$ is a record, we use the notation $\event{\id{e}\!}{\!\!G_{\id{e}}(\bar{x})\!\!}{\!\!\bar{x} := U_e(\bar{x})}$ to specify events. \christoph{For example, the event $\event{\id{dec}(a)\!}{\!\!z > a\!\!}{\!\!z := z - a}$ decreases $z$ by the parameter $a$ provided that the guard $z > a$ holds.}

\subsubsection{Refinement}
\label{ssec:refinement}

Given two event systems, $\ES_i = (S_i, E_i, \trans{}_i)$ and sets of initial states $I_i\subseteq S_i$ for $i\in\{1,2\}$, we say that \emph{$(\ES_2, I_2)$ refines $(\ES_1, I_1)$ modulo a mediator function $\pi\!: E_2\fun E_1$}, written $(\ES_2, I_2) \refines{\pi} (\ES_1, I_1)$, if there is a simulation relation $R \subseteq S_1 \times S_2$ such that 
\begin{enumerate}
\item for each $s_2\in I_2$ there exists some $s_1 \in I_1$ such that $(s_1, s_2)\in R$, and
\item for all $s_1\in S_1$, $s_2, s_2' \in S_2$ and $e_2 \in E_2$ such that $(s_1, s_2) \in R$ and $s_2 \trans{e_2}_2 s_2'$ there exists some $s_1' \in S_1$ such that $s_1 \trans{\pi(e_2)}_1 s_1'$ and $(s_1', s_2') \in R$.
\end{enumerate}
This is standard forward simulation~\cite{DBLP:journals/iandc/LynchV95}, augmented with the mediator function $\pi$, which allows us to vary the events in a refinement.
We assume that all models \christoph{$\mathcal{E}$} in our developments include a special stuttering event $\skipp\in E$, defined by $s \trans{\skipp} s$; \christoph{consequently,}
the trace properties $\traces(\ES, I)$ are closed under the addition and removal of $\skipp$ to traces. Events that are added in a refinement step often refine $\skipp$. 
%


We prove a standard soundness theorem stating that refinement implies trace inclusion. This trace inclusion in turn preserves trace properties (modulo the mediator $\pi$). Here, $\pi$ is applied to each element of each trace and $\pi^{-1}(P_1)$ consists of all traces that map element-wise to a trace in $P_1$. 
\begin{theorem}
\label{thm:refinement-soundness}
$(\ES_2, I_2) \refines{\pi} (\ES_1, I_1)$ implies $\pi(\traces(\ES_2, I_2)) \subseteq \traces(\ES_1, I_1)$.
\end{theorem}

\begin{lemma}
\label{lem:property-preservation}
Suppose $\ES_1, I_1 \models P_1$ and $\pi(\traces(\ES_2, I_2)) \subseteq  \traces(\ES_1, I_1)$. Then $\ES_2, I_2 \models \pi^{-1}(P_1)$.
\end{lemma}
%

\christoph{For complex or multi-level refinements, it may be advisable to reformulate the intended property~$P_1$ at the concrete level as $P_2$ and prove that $\pi^{-1}(P_1) \subseteq P_2$, which implies $\ES_2, I_2 \models P_2$.}


\subsubsection{Parallel Composition}
\label{ssec:parallel-composition}

\christoph{
Given two event systems, $\ES_i = (S_i, E_i, \rightarrow_i)$ for $i\in\{1,2\}$, a set of events $E$, and a partial function $\chi\!: E_1 \times E_2 \map E$, we define their \emph{parallel composition} $\ES_1 \eparallel{\chi} \ES_2 = (S, E, \rightarrow)$, where $S = S_1 \times S_2$ and $(s_1, s_2) \trans{e} (s_1', s_2')$ iff there exist $e_1\in E_1$ and $e_2 \in E_2$ such that $\chi(e_1,e_2) = e$, $s_1 \trans{e_1} s_1'$, and $s_2 \trans{e_2} s_2'$.  
%
We define the \emph{interleaving} composition $\ES_1 \interleave \ES_2 = \ES_1 \interl \ES_2$, where $E = E_1 \uplus E_2$ and $\chiI(e_1, \skipp) = e_1$, $\chiI(\skipp, e_2)=e_2$, and $\chiI(e_1, e_2)=\bot$ if $\skipp \notin \{e_1, e_2\}$. 
}

\christoph{
We can also define a composition on sets of traces. For two trace properties $T_1$ and $T_2$ over events $E_1$ and $E_2$, a set of events $E$, and a partial function $\chi\!: E_1 \times E_2 \map E$, we define $\tau \in T_1 \eparallel{\chi} T_2$ iff there exist $\tau_1 \in T_1$ and $\tau_2 \in T_2$ such that $\len{\tau_1} = \len{\tau_2} = len(\tau)$ and, for $0 \leq i < \len{\tau}$, we have $\chi(\tau_1(i), \tau_2(i))= \tau(i)$. 
We can then prove the following composition theorem (Theorem~\ref{thm:composition}), which enables compositional refinement (Corollary~\ref{cor:compositional-refinement}), where we can refine individual components while preserving trace inclusion for the composed system.
Similar results existed previously (see, \eg, \cite{DBLP:conf/fmco/SilvaB10}), but we have generalized them and formalized them in Isabelle/HOL\@.}
\begin{theorem}[Composition theorem]
\label{thm:composition}
\christoph{
$
\traces(\ES_1 \eparallel{\chi} \ES_2, I_1 \times I_2) 
= \traces(\ES_1, I_1) \eparallel{\chi} \traces(\ES_2, I_2).
$}
\end{theorem}


\begin{corollary}[Compositional refinement]
\label{cor:compositional-refinement}
Suppose $\traces(\ES_i', I_i') \subseteq \traces(\ES_i, I_i)$ for \mbox{$i\in\{1,2\}$}. Then $\traces(\ES_1'  \eparallel{\chi} \ES_2', I_1' \times I_2')  \subseteq \traces(\ES_1 \eparallel{\chi} \ES_2, I_1 \times I_2)$. 
\end{corollary}



\subsection{I/O Specifications for Separation Logic}
\label{ssec:io-separation-logic}

To satisfy the versatility and expressiveness requirements stated in the introduction, we use a verification technique that works with \emph{any} separation logic that offers a few standard features. This approach supports a wide range of programming languages, program logics, and verification tools.

We build on the work by \citet{DBLP:conf/esop/Penninckx0P15}, which enables the verification of possibly non-terminating reactive programs that interact with their environment through a given set of I/O operations, corresponding to I/O library functions, using standard separation logic. 
They introduce an expressive assertion language for specifying a program's allowed I/O behavior; for example, one can specify sequential, non-deterministic, concurrent, and counting I/O behavior. 
This language can be encoded into any existing separation logic that offers standard features such as abstract predicates~\cite{Parkinson2005}. 
Consequently, our approach inherits the virtues of these logics, for instance, local reasoning and support for language features such as mutable heap data structures and concurrency (including fine-grained and weak-memory concurrency).
In particular, our approach leverages existing program verification tools for separation logic, such as VeriFast~\cite{DBLP:conf/nfm/JacobsSPVPP11} (for Java and C), Nagini~\cite{DBLP:conf/cav/Eilers018} (for Python), and GRASShopper~\cite{PiskacWZ13}, and benefits from the automation they offer.

\subsubsection*{Syntax}
\label{ssec:iosep-logic-syntax}

We assume a given set of (basic) \emph{I/O operations} $bio\in\Bio$ and countably infinite sets of values $v, w \in \Val$ and places $t, t' \in \Place$. The set of \emph{chunks} is defined by 
\christoph{
\[
\Chunk ::= \bio(t,v,w,t') \mid \token(t),
\]
where $\bio\in\Bio$, $t,t'\in\Place$, and $v,w\in\Val$.
}
We call a chunk of the form $\bio(t,v,w,t')$ an \emph{I/O permission} to invoke the operation $bio$ with output~$v$, whose \emph{source} and \emph{target} places are~$t$ and~$t'$, respectively, and which \emph{predicts} receiving the input value~$w$.
Note that input and output are from the perspective of the calling system component, not the environment: for example, $\id{send}(t_1, 12, 0, t_2)$ models a permission to send the value~$12$ (output) and a prediction that the obtained result will be $0$ (input). A chunk of the form $\token(t)$ is called a \emph{token} at place $t$. Intuitively, the places and I/O permissions form the nodes and edges of a multigraph. Allowed I/O behaviors are obtained by pushing tokens along these edges, which consumes the corresponding I/O permissions. 

The language of assertions, intended to describe multisets of chunks representing possibly non-terminating behavior, is co-inductively defined (indicated by the subscript $\corec$) by 
\[
   \phi ::=_{\corec} b \mid c \mid \phi_1 \sand \phi_2 \mid \exists v.\; \phi \mid \exists t.\, \phi,
\]
where $b \in \Bool$, $c \in \Chunk$, $\phi_1 \sand \phi_2$ is the separating conjunction, and the two existential quantifiers are on values $v \in \Val$ and places $t \in \Place$, respectively. In separation logic, chunks can be modeled using abstract predicates; all other assertions are standard.
In our Isabelle/HOL formalization, we use a shallow embedding of assertions. 
%
%
Disjunction is encoded using existential quantification. We borrow other constructs such as the conditional ``$\ifte{b}{\phi_1}{\phi_2}$'', variables, and functions operating on values from the meta-language, namely, Isabelle's higher-order logic.
We also call assertions \emph{I/O specifications} to emphasize their use as program specifications. 



\begin{example}
\label{ex:simple-IOspec}
The following I/O specification allows receiving an integer and subsequently sending the doubled value. 
\[
  \phi = \token(t) \sand (\exists x, t', t''.\, \id{recv}(t, x, t') \sand \id{send}(t', 2x, t'')).
\]
Since the input value $x$ is existentially quantified and unconstrained, there is no prediction about 
the value that will be received.
Here, we use I/O permissions performing only input ($\id{recv}$) or only output ($\id{send}$) instead of both. 
For such permissions, we elide the irrelevant argument, implicitly setting it to a default value like the unit  $\unit$. The single token points to the source place $t$ of $\id{recv}$.
\end{example}

Note that I/O specifications use places to determine the execution order of different I/O operations  without requiring specific program logic support beyond normal separation logic. 
For example, sequential composition and choice are expressed by using separate chunks that share source or target places.
Determining whether an I/O operation may be performed is therefore as simple as checking whether there is a permission that has a source place with a token.
Other approaches use \christoph{custom} specification constructs to express this and require \christoph{custom logics} (\eg,~\citet{DBLP:conf/cpp/Koh0LXBHMPZ19,DBLP:conf/ifm/OortwijnH19}).

\paragraph{Repeating behavior.} The co-inductive definition of assertions allows us to define formulas co-recursively. For consistency, Isabelle/HOL requires that co-recursive calls are \emph{productive}~\cite{DBLP:conf/esop/BlanchetteBL0T17}, \christoph{namely,} guarded by some constructor, which is the case for all of our co-recursive definitions.
For example, for a countable set of values $S$, we define the iterated separating conjunction $\sall v \in S.\, \phi$.
%
%
We can also co-recursively define possibly non-terminating I/O behavior. 

\begin{example}
\label{ex:IOspec-read-write}
\tk{
The assertion $\phi = \token(t) \sand \id{RS}(t,0)$ specifies the behavior of
repeatedly receiving inputs and sending their sum, as long as the received values are positive.
}
\[ 
\id{RS}(t, a)  =_\corec \exists z, t', t''.\, \id{recv}(t,z,t') \sand 
  \ifte{z > 0}{\id{send}(t', a + z,t'') \sand \id{RS}(t'', a + z)}{\tru}.
\]
Here, the parameters $t$ and $a$ of $\id{RS}$ represent the current state. Since this is a co-recursive definition, it includes the non-terminating behaviors where all received values are strictly positive.
\end{example}

\subsubsection*{Semantics}

Assertions have both a static semantics in terms of multisets of chunks and a transition semantics for which we have given an intuition above. This intuition suffices to understand our methodology. We therefore defer the definition of the formal semantics to Section~\ref{ssec:io-separation-logic-semantics}. 


\section{Igloo Methodology}
\label{sec:methodology}

In this section, we present our approach for developing fully verified distributed systems, which satisfies the requirements set out in the introduction. 
Our approach applies to any system whose components maintain a local state and exclusively communicate over an environment such as a communication network or a shared memory system.
There are neither built-in assumptions about the number or nature of the different system components nor about the environment. In particular, the environment may involve faulty or adversarial behavior.
We also support different programming languages and code verifiers for the implementation and the interoperability of heterogenous components written in different languages. This versatility is enabled by separating the modeling and implementation side and using I/O specifications to link them. 
%

After giving an overview of our methodology (Section~\ref{ssec:overview-formal-steps}) \christoph{and the distributed leader election protocol case study (Section~\ref{ssec:case-study})}, we explain \pmue{our methodology's steps and illustrate them by transforming} an informal, high-level description of the system and its environment into real-world implementations in Java and Python with formal correctness guarantees (Sections~\ref{ssec:abstract-phase}--\ref{ssec:component-implementation-and-verification}).
We summarize \christoph{our approach's} soundness arguments (Section~\ref{ssec:summary-and-guarantees}): trace properties established for the models are preserved down to the implementation provided that our trust assumptions (Section~\ref{ssec:trustassumptions}) hold.
We currently support the verification of safety properties, but not liveness properties.

\subsection{Overview of Formal Development Steps}
\label{ssec:overview-formal-steps}

Before we start a formal development, we must identify the \christoph{system requirements} and the assumptions about the environment.
The system requirements include the (informally stated) goals to be achieved by the system and structural constraints such as the types of its components. 
The environment assumptions describe the properties of the environment, including communication channels (\eg, asynchronous, lossy, reordering channels), the types of component faults that may occur (\eg, crash-stop or Byzantine~\cite{DBLP:books/daglib/0025983}), and possible adversarial behavior (\eg, the Dolev-Yao model of an active network adversary~\cite{DBLP:journals/tit/DolevY83}). 
%

Our methodology consists of six steps (\cf~Figure~\ref{fig:approach}). In Steps~1--2, we use standard refinement to develop a detailed model of the system and its environment. The number of refinements per step is not fixed.  
Each refinement is proven correct and may incorporate additional system requirements.

%
\begin{enumerate}
\item \emph{Abstract models}. We start with an abstract model that takes a global perspective on the problem. It may solve the problem in a single transition. Typically, the most central system properties are already established for this model, or the abstract models that further refine it.

\item \emph{Protocol models}. We then move from the global to a distributed view, where nodes execute a protocol and communicate over the environment. The result of this step is a model that incorporates all system requirements and environment assumptions.
\end{enumerate}
 

In Steps~3--6, we produce an interface model from which we can then extract component specifications, implement the components, and verify that they satisfy their specifications.
\begin{enumerate} \setcounter{enumi}{2}
\item \emph{Interface models}. We further refine the protocol model for the subsequent decomposition into system components and the environment, taking into account the I/O library interfaces to be used by the implementation. 

\item \emph{Decomposition.} We decompose the monolithic interface model into system components and the environment. Their re-composition is trace-equivalent with the monolithic model.

\item \emph{Component I/O specification.} We translate the component models into trace-equivalent I/O specifications (in separation logic) of the programs that implement the components. 

\item \emph{Component implementation and verification.} We implement the components in a suitable programming language and prove that they satisfy their I/O specification. 
\end{enumerate}

Steps~1--4 are supported by a generic refinement and composition framework that we have embedded in Isabelle/HOL (see Sections~\ref{ssec:event-systems-refinement} and~\ref{ssec:decomposition}). 
%
Steps~3--5 are novel and specific to our approach. In Steps~3 and~4, we align our models' events with the implementation's I/O library functions and then separate the interface model into a set of possibly heterogeneous system components (\eg, clients and servers) and the application-specific environment (\eg, modeling a particular network semantics, faulty, or adversarial behavior).
Step~5 constitutes one of the core contributions of our approach: a sound link between abstract models and I/O specifications in separation logic, also formalized in Isabelle/HOL\@. It will be introduced informally here and formalized in Section~\ref{sec:theory}. 
Step~6 corresponds to standard code verification, using tools such as Nagini (for Python) and VeriFast (for Java and C). Due to our clear separation of modeling and implementation, the code verifier must  check only that a component implementation follows the protocol; 
code verification neither needs to reason about the protocol's global properties nor about the environment, which simplifies verification and increases modularity.
%
\christoph{
In Section~\ref{ssec:summary-and-guarantees}, we will derive the overall soundness of our methodology from the individual steps' guarantees, which are summarized in Table~\ref{tab:steps-and-guarantees}.
}

Our three case studies demonstrate the versatility and expressiveness of our approach. We cover different types of systems, including fault-tolerant and secure ones, different component types with unbounded numbers of instances, and TCP and UDP communication. We have written and verified implementations in Python and Java, including concurrent ones. 
This section illustrates our approach using the leader election case study; the other case studies are presented in Section~\ref{sec:case-studies}.


\subsection{Case Study: Leader Election}
\label{ssec:case-study}

%
%
The main requirement of a distributed leader election protocol is to elect at most one leader in a network of uniquely identified but otherwise identical nodes, whose total number is a priori unknown. Since we do not consider liveness properties in this work, we do not prove that the protocol will terminate with an elected leader. 


We model an algorithm by~\citet{DBLP:journals/cacm/ChangR79}, which assumes a ring network and a strict total order on the set of node identifiers.
The algorithm elects the node with the maximum identifier as follows. Each node initially sends out its identifier to the next node in the ring and subsequently forwards all received identifiers greater than its own. When a node receives its own identifier, this is guaranteed to be the maximum identifier in the ring, and the node declares itself the leader.
For the environment, we assume that each node asynchronously sends messages to the next node in the ring over an unreliable, duplicating, and reordering channel. We do not consider other faults or adversarial behavior 
in this example, but see Section~\ref{sec:case-studies} for case studies that do.

\subsection{Step 1: Abstract Models}
\label{ssec:abstract-phase}

A common approach to develop systems by refinement is to start from a very abstract model whose correctness is either obvious or can be proved by a set of simple invariants or other trace properties. 
This model takes a global ``bird's eye'' view of an entire run of the protocol in that it does not explicitly model the network communication or represent the individual protocol steps.
\christoph{
\begin{example}
\label{ex:leader-election-abstract}
The abstract model of leader election elects a leader in a single ``one-shot'' transition. We assume a given set $\Nodeid$ of node identifiers. 
The model's state space is defined as an $\Nodeid$-indexed family of local states containing a single boolean state variable identifying the leader, \ie, $S_0 = \Nodeid \fun \record{\winner \in \Bool}$.  Initially, $\winner(s_0(i)) = \fal$, for all $i \in \Nodeid$. There is  a single event $\mathit{elect}$ that elects the leader. 
The guard ensures that this event can be performed only by a single, initially arbitrary node that updates its local variable $\winner{}$ to $\tru$.
\[
   \event{\id{elect}(i)}{(\forall j.\; \winner_j \Rightarrow i = j)}{\winner_i := \tru}.
\]
We use indexing to refer to different instances of variables, \eg, $\winner_j$ refers to node $j$'s local state. Note that the guard refers to other nodes' local states; hence, this model takes a global point of view.
We have proved that this model satisfies the main requirement for leader election, namely, the uniqueness of the leader. This is formalized as the trace property
\[
   U_0 = \{\tau \mid \forall i, j.\,\id{elect}(i) \in \tau \land \id{elect}(j) \in \tau \Rightarrow i = j\},
\] 
where $e\in \tau$ means that the event $e$ occurs in the trace $\tau$.
%
This model is sufficiently abstract to specify \emph{any} leader election algorithm, and will be refined to the protocol described above next.
\end{example}
}

\subsection{Step 2: Protocol Models}
\label{ssec:protocol-phase}

In Step~2, we move from a global to a distributed perspective, and distinguish system components (\eg, nodes or clients and servers) that communicate over an environment (\eg, \christoph{a wide-area network}). The way that we model the environment accounts for any assumptions made about network communication. For example, we can represent a reliable, non-duplicating, reordering channel as a multiset of messages. This step may also introduce a failure model for fault-tolerant systems or an adversary model for secure systems. 
%
The result of this step is a complete model of our system and environment that satisfies all system requirements. 

\begin{example}
\label{ex:leader-election-protocol}
\christoph{We refine our} abstract model into a protocol model. We model the environment by assuming a finite, totally ordered set of identifiers $\Nodeid$ and that the nodes are arranged in a ring defined by a function $\nextt \!: \Nodeid \rightarrow \Nodeid$, where $\nextnode{i}$ yields node $i$'s successor in the ring.
%
We extend the state with communication channels, which we model as sets, from which messages are never removed; this represents our assumption that the network may reorder and duplicate messages. Since we do not consider liveness properties, message loss is implicitly represented by never receiving a message.
Since messages contain node identifiers, our state space becomes $S_1 = \Nodeid \fun \record{\winner\in\Bool, \chan \! \in \Pow(\Nodeid)}$.  

Three events model the protocol steps: a $\id{setup}$ event where nodes send their own identifier to the next node in the ring, an $\id{accept}$ event where they forward received identifiers greater than their own, and an $\id{elect}$ event where a node receiving its own identifier declares itself the leader.
\[
\begin{array}{lll}
   \event{\id{setup}(i)&}%
            {\tru&}%
            {\chan_{\nextnode{i}} := \chan_{\nextnode{i}} \cup \{i\}} 
   \\
   \event{\id{accept}(i, j)&}%
             {j \in \chan_i \land j > i&}%
             {\chan_{\nextnode{i}} := \chan_{\nextnode{i}} \cup \{j\}} 
   \\
   \event{\id{elect}(i)&}%
             {i \in \chan_i &}%
             {\winner_i := \tru}
\end{array}
\]

\christoph{
We have proved that this protocol model 
refines the abstract model defined in Example~\ref{ex:leader-election-abstract}. For this we use the simulation relation that removes the field $\chan$ from the local state and the mediator function that maps $\id{elect}$ to itself and the new events to $\skipp$. 
The proof involves showing that the guard of this model's $\id{elect}$ event implies the guard of the abstract model's $\id{elect}$ event. 
We prove two invariants that together imply this. The first one is inductive and states that if a node ID $i$ is in the channel of node $j$ then $k < i$ for all node IDs $k$ in the channels in the ring interval from $i$ to $j$. From this it follows that if $i \in \chan_i$ then $i$ is the maximal node ID. The second invariant expresses that only the node with the maximal node ID can become a leader.
}
\end{example}

\subsection{Step 3: Interface Models}
\label{ssec:interface-phase}

\christoph{This is the first step towards an implementation. Its purpose is twofold: first, we prepare the model for the subsequent decomposition step (Step 4) and, second, we align the I/O events with the API functions of the I/O libraries to be used in the implementation.}
The resulting interface model must satisfy the following \christoph{structural} \emph{interface requirements}: 
\begin{enumerate}
\item The state space is a product of the components' local state spaces and the environment's state space. The events are partitioned into \christoph{\emph{I/O events}, which model the communication with the environment, and \emph{internal events}, which model local computations.
} 

\item Each I/O event can be associated with a single I/O library function (\eg, receiving or sending a message on a socket, but not both). It must have the same parameters as that library function, each of which can be identified as \christoph{an output parameter (\eg, the message to send) or an input parameter (\eg, an error code returned as a result).} 


\item 
Each I/O event's guard must be the conjunction of 
\begin{itemize}
\item 
a \emph{component guard}, which refers only to the component's local state, the event's output parameters, \christoph{and the component identifier}, and 
\item 
an \emph{environment guard}, referring only to the environment's state, the input parameters, \christoph{and the component identifier}. 
\end{itemize}
\end{enumerate}

Our approach leaves the choice of the abstraction level of the interface model's I/O events to the user. For example, the APIs of network socket libraries typically represent payloads as bitstrings, which the application must parse into and marshal from its internal representation. We may choose to either (i) define I/O events (and thus I/O operations) that operate on bitstrings, which requires modeling and verifying  parsing and marshalling explicitly, or (ii) keep their interface on the level of parsed data objects, and trust that these functions are implemented correctly.

\begin{example}
\label{ex:interface-model}
We refine the protocol model into a model satisfying the interface requirements. The protocol model's $\id{accept}$ event receives, processes, and sends a message. To satisfy Conditions~1--2 above, we introduce two local buffers, $\ibuf$ and $\obuf$, for each node and split $\id{accept}$ into three events: 
$\id{receive}$ transfers a message from the previous node to the input buffer $\ibuf$, 
$\id{accept}$ processes a message from $\ibuf$ and places the result in the output buffer $\obuf$, and
$\id{send}$ sends a message from $\obuf$ to the next node.
%

We also align the I/O events $\id{send}$ and $\id{receive}$ with the I/O operations \christoph{$\id{UDP\_send\_int}(\id{msg}, \id{addr})$}
and $\id{UDP\_receive\_int}(\id{msg})$, which are offered by standard socket libraries. 
%
Here, we represent messages as integers, but as stated above, we could alternatively represent them as bitstrings, and model parsing and marshalling explicitly (including 
bounds and endianness), resulting in stronger correctness guarantees. 
Since each I/O event must match the corresponding I/O operation's parameters (Condition~2), we add the send operation's destination address as an event parameter. Hence, we introduce an injective function $\addrNA : \Nodeid \rightarrow \Addr$, where $\Addr$ is the set of addresses. 
UDP communication is unreliable and messages sent may be reordered, duplicated, or lost; our environment model faithfully represents this behavior by modeling channels as sets (Section~\ref{ssec:protocol-phase}).%

\begin{figure}[t]
\begin{center}
\[
\begin{array}{lll}
   \event{\christoph{\id{setup}_i()} &}%
            {\tru &}%
            {\obuf_{\!i} := \obuf_{\!i} \cup \{i\}} 
   \\
   \event{\christoph{\id{receive}_i(j)} &}%
            {j \in \chan_{\addr{i}} &}%
            {\ibuf_{\!i} := \ibuf_{\!i} \cup \{j\}} 
   \\
   \event{\christoph{\id{accept}_i(j)} &}%
            {j \in \ibuf_{\!i} \land j > i &}%
            {\obuf_{\!i} := \obuf_{\!i} \cup \{j\}}
   \\
   \event{\christoph{\id{send}_i(j, a)} &}%
            {j \in \obuf_{\!i} \land a = \addr{\nextnode{i}} &}%
            {\chan_{a} := \chan_{a} \cup \{j\}}
   \\
   \event{\christoph{\id{elect}_i()} &}%
            {i \in \ibuf_{\!i} &}%
            {\winner_i := \tru}.
\end{array}
\]
\vspace{-4mm}
\caption{Event system resulting from interface refinement step.}
\label{fig:leader-election-interface-model}
\end{center}
\vspace{-2mm}
\end{figure}

We define the state space as the product $S_2 = \system{S}_2 \times \environ{S}_2$ (Condition~1) of a system state space $\system{S}_2 = \Nodeid \fun \record{\winner \in \Bool, \ibuf \in \Powerset{\Nodeid}, \obuf \in \Powerset{\Nodeid}}$ and an environment state space $\environ{S}_2 = \Addr \fun \record{\chan  \in \Powerset{\Nodeid}}$. 
%
The events are specified in Figure~\ref{fig:leader-election-interface-model}. 
\christoph{
We henceforth consider the component identifier $i$ as a component parameter and therefore write it as a subscript of the event.
}
Only $\id{receive}$ and $\id{send}$ are I/O events; all others are internal (Condition~1). These I/O events have the required form and parameters (Condition~2) and their guards have the required separable form (Condition~3). The parameter \christoph{$j$} of $\id{receive}$ is the only input parameter and all others are outputs.
The simulation relation with the protocol model projects away the internal buffers. The mediator function maps $\id{elect}$ to itself, $\christoph{\id{send}_i(j, a)}$ to $\christoph{\id{setup}(i)}$ if $i=j$ and to $\christoph{\id{accept}(i,j)}$ otherwise, and all other events to $\skipp$.
The refinement proof requires an invariant relating internal buffers to channels, \eg, stating that $\christoph{j} \in \id{ibuf_{i}}$ implies $\christoph{j} \in \id{chan}_{\addr{i}}$.
\end{example}
%



\subsection{Step 4: Decomposition}
\label{ssec:decomposition}

To support distributed systems with different component types (such as nodes or clients and servers), \christoph{we decompose the monolithic interface model from Step~3 into a parallel composition of an environment model and (a family of) component models for each component type.} 


%

%
\christoph{
We first decompose the interface model into a parallel composition $\mathcal{E} = \Systmodel \eparallel{\chi} \Envmodel$ of a system model $\Systmodel$ and an environment model $\Envmodel$. We have already distinguished their respective state spaces $\system{S}$ and $\environ{S}$ in the interface model.
The I/O events $e$ of $\mathcal{E}$ are split into a system part $\system{e}$, consisting of $e$'s component guard and system state updates, and an environment part $\environ{e}$, consisting of $e$'s environment guard and environment state updates. We define $\chi$ such that it synchronizes the split I/O events and interleaves the internal events.
}
The system model is \christoph{further} subdivided into models of different component types, which are composed using interleaving composition $\Systmodel = \Systmodel_1 \interleave \cdots \interleave \Systmodel_n$. This reflects our assumption that the components exclusively communicate via the environment. If there are multiple \emph{instances} of a component type, parametrized by a countable index set $I$ of identifiers, the respective model, say $\Systmodel_k$, becomes an interleaving composition over~$I$, 
\christoph{
that is, $\interleave_{i\in I}\: \Systmodel_k(\vec{\gamma}_k(i))$. Each component model $\Systmodel_k(\vec{p})$ may have some parameters $\vec{p}$. We instantiate these using a \emph{configuration map} $\vec{\gamma}_k$, which represents assumptions on the \emph{correct system configuration}.
}%
%
%
Note that component models may be further refined before translating them to I/O specifications. 


In preparation for the subsequent translation to I/O specifications, we model (instances of) system components in a subclass of guarded event systems. An \emph{I/O-guarded event system} $\GES = (S, E, G, U)$ is a guarded event system, where $E$ consists of events of the form $\bio(v,w)$ (formally introduced as \emph{I/O actions} in Section~\ref{ssec:heap-transitions}) 
and all guards $G_{\bio(v,w)}$ are component guards as in Condition~(3), \ie, they \pmue{must not depend on the I/O action's input~$w$}. This models that an input becomes available only as the result of an I/O operation and cannot be selected before the I/O operation is invoked. 
Furthermore, we model a component's internal events as \emph{ghost I/O actions}; \christoph{these actions} change the state of the abstract model, but do not correspond to real I/O operations. The implementation may have to perform a corresponding state change to stay aligned with the abstract model. 

We prove the correctness of the decomposition by showing that the parallel (re)composition of all parts is trace-equivalent to the original system.



\begin{example}
\label{ex:decomposition}
All nodes instantiate the same component type. We thus decompose the model from the previous step into an environment event system $\Envmodel$ and an I/O-guarded event system $\Systmodel(i, a)$, parametrized by a node identifier $i \in \Nodeid$ and an address $a \in \Addr$.
These will also be the parameters of the future program $c(i, a)$ implementing $\Systmodel(i, a)$. 
\christoph{
For the system's (re)composition, we use the configuration map $\vec{\gamma}(i)=(i,\addr{\nextnode{i}})$, which instantiates the destination address $a$ for $i$'s outbound messages with the address of $i$'s successor in the ring.
}
%
The environment operates on the state \christoph{$\environ{S}_3 = \Addr \fun \record{\chan \in \Powerset{\Nodeid}}$} and the state space of each node model \christoph{$\Systmodel(i,a)$ is $\system{S}_3 = \record{\winner \in \Bool, \ibuf\in \Powerset{\Nodeid}, \obuf\in \Powerset{\Nodeid}}$}. The environment has the following events, where `$-$' represents the identity update function:
\[
\begin{array}{lll}
   \event{\environ{\id{receive}}(i, m) &}%
            {m \in \chan_{\addr{i}} &}%
            {-} 
   \\
   \event{\environ{\id{send}}(i, m, a) &}%
            {\tru &}%
            {\chan_{a} := \chan_{a} \cup \{m\}}.
\end{array}
\]
These events execute synchronously with their matching system parts:
\[
\begin{array}{lll}
   \event{\system{\id{receive}}_{i,a}(m) &}%
            {\tru &}%
            {\ibuf{} := \ibuf{} \cup \{m\}} 
   \\
   \event{\system{\id{send}}_{i,a}(m, a') &}%
            {m \in \obuf \land a' = a &}%
            {-}.
\end{array}
\]

\noindent
Note that the $\system{\id{receive}}$ event's guard does not depend on its input parameter $m$ and the $\system{\id{send}}_{i,a}$ event's single output parameter
is a pair of a message and an address. \christoph{The equality $a'=a$ in \pmue{the guard of} $\system{\id{send}}_{i,a}$ enforces that messages are sent only to the node at the address $a$, which is a component parameter. This is a constraint on the future program's use of the I/O library function.}
The internal events $\id{setup}_{i,a}()$, $\id{accept}_{i,a}(m)$, and $\id{elect}_{i,a}()$ of $\Systmodel(i, a)$ are ghost I/O actions, which are identical to their counterparts in the previous model modulo their slightly different parametrization. We have proved that the composition of all parts is trace-equivalent to the original monolithic system.
\end{example}

\subsection{Step 5: I/O Specifications}

We can now perform the central step of our approach: we extract, for each component, an I/O specification that defines the implementation's I/O behavior. Our translation maps an I/O-guarded \christoph{parametrized} event system $\Systmodel(\vec{p})$ to an I/O specification of the form 
\begin{align*}
	\phi(\vec{p}) = \exists t.\, \token(t) \sand P(t, \vec{p}, s_0),
\end{align*}
where $P$ is a co-recursively defined predicate encoding the events and $s_0$ is the event system's initial state.\footnote{\christoph{The formal development of our theory (Section~\ref{sec:theory}) is based on event systems with single initial states. This is without loss of generality since multiple initial states can easily be introduced by a non-deterministic initialization event.}} 
The predicate $P$ takes a place $t$, the event system's (and future program's) parameters~$\vec{p}$, and the event system's abstract state $s$ as arguments. The predicate $P$ contains, for each event and all values of its output parameters satisfying the guard, a permission to execute the I/O operation represented by the event, and another instance of itself with the argument representing the new state resulting from applying the event's update function. 
This translation
is formally defined and proved correct in Section~\ref{sec:theory}. 
Here, we explain the intuition behind it using our example. 

\begin{figure}[t]
\begin{center}
\begin{align*}
	P(t, (i, a), s) =_\corec &~(\exists t'.\, \id{setup}(t, t') \sand P(t', (i, a), s\record{\obuf{} := \obuf(s) \cup \{i\}})) \sand {} \\
	&~(\exists m, t'.\, \id{UDP\_receive\_int(t, m, t') \sand P(t', (i, a), s\record{\ibuf{} := \ibuf(s) \cup \{m\}})}) \sand {} \\
	&~(\sall m.\, \ifte{m \in \ibuf(s) \wedge i < m}{\exists t'.\, \id{accept}(t, m, t') \sand {} \\
	& \phantom{~(\sall m.\, \exists t'.\, } P(t', (i, a), s\record{\obuf{} := \obuf(s) \cup \{m\}})}{\tru} ) \sand {}\\
	&~(\sall m, a'.\, \ifte{m \in \obuf(s) \land a' = a}{\exists t'.\, \id{UDP\_send\_int}(t, (m, a'), t') \sand {} \\
	&\phantom{~(\sall m, a'.\,} P(t', (i, a), s)}{\tru} ) \sand {}\\
	&~ (\ifte{i \in \ibuf(s)}{\exists t'.\, \id{elect}(t, t') \sand {} 
	P(t', (i, a), s\record{\winner{} := \tru})}{\tru}).
\end{align*}
\vspace{-4mm}
\caption{I/O specification of leader election nodes.}
\label{fig:io-spec-leader-election}
\end{center}
\vspace{-2mm}
\end{figure}

\begin{example}
\label{ex:iospec}
Figure~\ref{fig:io-spec-leader-election} defines the predicate $P(t, (i,a), s)$ for our example, where $i$ and $a$ denote the local node identifier $i$ and the address $a$ of the next node in the ring. 
The fourth top-level conjunct of $P$ corresponds to the $\system{\id{send}}_{i,a}(m, a')$ event from the previous step. It states that for all possible values of the output parameter $(m,a')$ that fulfill the event's guard $m \in \obuf(s) \land a' = a$, there is a permission to perform the I/O operation $\id{UDP\_send\_int}$ (which is mapped to the $\system{\id{send}}_{i,a}$ event) and another instance of $P$ at the operation's target place with the same state, since $\system{\id{send}}_{i,a}$ does not change the local state. 
The second (simplified) conjunct corresponds to the $\system{\id{receive}}_{i,a}$ event and 
existentially quantifies over the event's input parameter and contains another predicate instance with an updated state $s\record{\ibuf{} := \ibuf(s) \cup \{m\}}$ as defined by $\system{\id{receive}}_{i,a}$.
\christoph{The remaining conjuncts correspond to the internal events $\id{setup}$, $\id{accept}$, and $\id{elect}$.}
\end{example}

\subsection{Step 6: Component Implementation and Verification}
\label{ssec:component-implementation-and-verification}

In the final step, we prove for every component that its implementation $c$ fulfills the I/O specification~$\phi$ that was extracted in the previous step. This requirement is expressed as
\begin{equation}
\traces(c) \subseteq \traces(\phi), 
\label{eq:program-traces-cond}
\end{equation}
\ie, the I/O traces of the component implementation $c$, as defined by 
its operational semantics, are included in those specified by the I/O specification~$\phi$. 
\christoph{Here, we elide possible parameters $\vec{p}$ of the program $c$ and the I/O specification $\phi$ for the sake of a lighter notation.}
Since I/O specifications are language-agnostic, the implementation may use any programming language. Verifying~\eqref{eq:program-traces-cond} typically requires defining a suitable I/O-aware semantics of the chosen language that defines the I/O traces produced by its programs. We assume that the verification technique used defines an interpretation of Hoare triples of the form $\models \Hoare{\phi}{c}{\psi}$, and a sound program logic to prove them.
We only rely on the \emph{verifier assumption} 
\christoph{stating that the correctness of a command with respect to a precondition implies the trace inclusion between the command and the precondition assertion.}
Since the I/O permissions in the precondition restrict which I/O operations may be performed, these triples do not trivially hold even though the postcondition is $\tru$:
\begin{equation}
\models \Hoare{\phi}{c}{\tru} \;\text{ implies }\; \traces(c) \subseteq \traces(\phi).
\tag{VA}\label{eq:verifier-assumption}
\end{equation}

Our approach leaves open the mechanism for proving such Hoare triples. In principle, proofs can be constructed using an interactive theorem prover, an SMT-based automated verifier, or even as pen-and-paper proof.  
I/O specifications consist only of constructs that can be expressed using standard 
separation logic with abstract predicates. This allows us to leverage existing tool support, in particular, proof automation. For instance, encoding such specifications in VeriFast required less than 25 LoC to declare types for places and abstract predicates for chunks, and no tool modifications.

\begin{figure}[t]
\begin{center}
\begin{lstlisting}[basicstyle=\scriptsize\sffamily,language=Python,caption={Pseudocode of the leader election algorithm with proof annotations (simplified). The method $\id{try\_receive\_int}$ tries a receive operation and either returns an identifier or times out and returns $\id{None}$.},captionpos=b,label={lst:impl}]
def main(my_id: int, out_host: str):
    #@ PRE: exists t . token(t) and P(t, (my_id, out_host), INIT_STATE)
    #@ POST: true
    to_send = my_id    # variable stores only the largest identifier seen so far
    setup()            # ghost I/O operation
    while True:
        #@ INVARIANT: exists t, s . token(t) and P(t, (my_id, out_host), s)
        #@ INVARIANT: to_send in s.obuf and to_send >= my_id
        send_int(out_host, to_send)
        msg = try_receive_int()     # returns None on timeout
        if msg is not None:
            if msg is my_id:
                elect()    # ghost I/O operation
                break
            elif msg > to_send:
                accept(msg)    # ghost I/O operation
                to_send = msg
\end{lstlisting}
\end{center}
\end{figure}

The I/O specification is (currently manually) converted to the syntax of the respective tool, and the program verifier is then used to prove the correctness of the program with respect to its I/O specification. Assuming~\eqref{eq:verifier-assumption} holds for the tool, this guarantees the required trace inclusion~\eqref{eq:program-traces-cond}.

\christoph{Besides the verification, we must manually justify that the actual program's I/O operations satisfy the assumptions encoded in the environment model. For example, we may implement an order-preserving network channel model using TCP sockets, but not UDP sockets. Conversely, it is sound to implement an unordered channel model using either TCP or UDP communication.}

\begin{example}
\label{ex:implementation}
We have implemented three versions of the leader election algorithm, a sequential and a concurrent one in Java and a sequential version in Python, and verified in VeriFast and Nagini that these implementations conform to their I/O specification:
$\models \Hoare{\phi(i, a)}{\id{main}(i, a)}{\tru}$.
All three implementations are interoperable and successfully elect a leader in actual networks. 

Listing~\ref{lst:impl} shows a slightly simplified \christoph{pseudocode version of the sequential implementation; the Java and Python versions share} the same structure but contain additional annotations as required by the respective verifier. The concurrent version uses two separate threads for receiving and sending identifiers.
\christoph{We use the standard UDP socket libraries of the} respective languages; since the API in both cases is structured differently, we defined the I/O operations used in the specification to be compatible with both. We annotated the relevant I/O library operations with contracts, whose correctness is assumed and must be validated manually \christoph{against the environment model}. 

\begin{figure}[t]
\begin{center}
\begin{lstlisting}[basicstyle=\scriptsize\sffamily,language=Python,caption={The simplified pseudocode contract for a library method for sending packets via UDP. The names starting with question marks are implicitly existentially quantified. \textit{connected} is a separation logic predicate that contains the heap memory of a UDP socket object connected to the shown address.},captionpos=b,label={lst:send_stub}]
def send_int(address: str, msg: int):
    #@ PRE: token(?t) * UDP_send_int(t, (msg, address), ?tp)
    #@ PRE: connected(this, address)
    #@ PRE: 0 <= msg <= MAX_MSG_VAL
    #@ POST: token(tp)
    #@ EXCEPTIONAL POST: token(t) * UDP_send_int(t, (msg, address), tp)
\end{lstlisting}
\end{center}
\end{figure}

\christoph{%
Listing~\ref{lst:send_stub} shows a simplified pseudocode specification of a message sending function. Its precondition consists of three typical parts that
(i) specifiy the I/O behavior of this function in terms of tokens and I/O permissions, 
(ii) constrain the program state, in this case requiring that our socket is already connected to the receiver address, and  
(iii) impose additional restrictions on messages that do not exist on more abstract levels, in this case, that the sent message falls within a valid range.
}


Since the I/O specification describes the allowed I/O behavior in terms of the model's state, the verification process requires relating the program to the model state. The latter is represented in the program as a \emph{ghost state}, which is present only for verification purposes, but not during program execution. 
If the verifier can prove for a given program point that a token for a place~$t$ and the predicate $P(t, (i, a), s)$ are held 
for some model state $s$, this means that the current program state corresponds to the model state $s$. The invocations of the internal operations $\id{setup}$, $\id{accept}$, and $\id{elect}$ in the code update the ghost state to stay aligned with the program state.

As an optimization, the implementations store and forward only the largest identifier seen so far, since smaller ones can never belong to the leader. The verifier proves the loop invariant that this largest integer is always in the output buffer and may therefore be sent out.
%

Note that, although we do not prove liveness, our implementation repeatedly resends UDP packets since packets may be lost. This will continue even after a leader has been elected since our simple protocol does not include a leader announcement phase.
\end{example}

\subsection{Overall Soundness Guarantees}
\label{ssec:summary-and-guarantees}

Our methodology provides a general way of proving properties of a distributed system. Table~\ref{tab:steps-and-guarantees} summarizes the soundness guarantees of each step (see also Figure~\ref{fig:approach}).  
We now show how to combine them to obtain the overall soundness guarantee that \christoph{the models' properties are preserved down to the implementation.}
\christoph{
We will first discuss the simpler case with a single instance of each component and later extend this reasoning to multiple instances.
}
%
%

\begin{table}[t]
\caption{Method overview with guarantees of each step (initial states elided).}
\label{tab:steps-and-guarantees}
\vspace{-5pt}
\begin{center}
\begin{tabular}{c|l|l|l}
steps  & activity & guarantee & justification  
\\ \hline
1--3   
& model refinements and
& $\hat{\pi}_i(\traces(\mathcal{M}_m)) \subseteq \traces(\mathcal{M}_i)$
& ref.~$\mathcal{M}_{i+1} \refines{\pi_{i+1}} \mathcal{M}_{i}$
\\ 
& interface refinement
& where $\hat{\pi}_i = \pi_m \circ \cdots \circ \pi_{i+1}$
& Theorem~\ref{thm:refinement-soundness}
\\ \hline
4          
& decompose $\mathcal{M}_m$ into 
& $\traces(\mathcal{M}_m) = $
& mutual refinement   
\\
& $\Systmodel_1, \ldots, \Systmodel_n$, and $\Envmodel$ 
& $\traces((\Systmodel_1 \interleave \ldots \interleave \Systmodel_n) \eparallel{\chi_{e}} \! \Envmodel)$
& Theorem~\ref{thm:refinement-soundness}
\\ \hline
5          
& translate $\Systmodel_j$ into $\phi_j$
& $\traces(\phi_j) = \traces(\Systmodel_j)$
& Theorems~\ref{thm:ioges-into-proc-correct} and~\ref{thm:gorilla-glue}
\\ \hline
6      
& verify $\models \Hoare{\phi_j}{c_j}{\tru}$
& $\traces(c_j) \subseteq \traces(\phi_j)$
& Assumption~\eqref{eq:verifier-assumption}
\end{tabular}
\end{center}
\end{table}

Let our implemented system be defined by
$
\mathcal{S} = (\mathcal{C}_1 \interleave \ldots \interleave \mathcal{C}_n) \eparallel{\chi_e}\!\Envreal, 
$
where each $\mathcal{C}_j$ is the event system defining the operational semantics of $c_j$, \ie, 
$\traces(\mathcal{C}_j) = \traces(c_j)$ and suppose the event system~$\Envreal$ corresponds to the real environment. 
From Steps~5--6's guarantees listed in Table~\ref{tab:steps-and-guarantees}, we derive
$
\traces(\mathcal{C}_j) \subseteq \traces(\Systmodel_j).
$    
\christoph{Furthermore, we assume that  the environment model faithfully represents the real environment, \ie,}
\begin{equation}
\tag{EA}\label{eq:environment-assumption}
\christoph{\traces(\Envreal) \subseteq \traces(\Envmodel).}
\end{equation}

Using Corollary~\ref{cor:compositional-refinement} and Step~4's guarantee, we derive that the implemented system's traces are included in the interface model's traces:
$
\traces(\mathcal{S}) \subseteq \traces(\mathcal{M}_m).
$  
\christoph{
Using Lemma~\ref{lem:property-preservation} and the guarantees of Steps~1--3, we derive our overall soundness guarantee stating that any trace property $P_i$ of model $\mathcal{M}_i$ is preserved all the way down to the implementation:
}
\begin{equation}
\label{eq:igloo-sound}\tag{SOUND}
\christoph{
\mathcal{M}_i \models P_i \;\Longrightarrow\; \mathcal{S} \models \hat{\pi}_i^{-1}(P_i).
}
\end{equation}

\christoph{%
With multiple component instances, the event systems $\Systmodel_j$ and $\mathcal{C}_j$ have the form of compositions $\interleave_{i\in I}\: \Systmodel_j(\vec{\gamma}_j(i))$ and  $\interleave_{i\in I}\: \mathcal{C}_j(\vec{\gamma}_j(i))$ for some configuration map $\vec{\gamma}_j$ and we have $\traces(\mathcal{C}_j(\vec{p})) = \traces(c_j(\vec{p}))$ for all parameters $\vec{p}$. The guarantees of Step~4 in Table~\ref{tab:steps-and-guarantees} hold for these compositions and those of Steps~5--6 for the individual parametrized components. We then easily derive the guarantee~\eqref{eq:igloo-sound} using a theorem for indexed interleaving composition similar to Corollary~\ref{cor:compositional-refinement}.  
}

\subsection{Trust Assumptions}
\label{ssec:trustassumptions}

The guarantees described in the last subsection hold under the following trust assumptions:

\begin{enumerate}

	\item \emph{Environment assumptions:} The modeled environment includes all possible behaviors of the real system's environment, \christoph{as formulated in Assumption~\eqref{eq:environment-assumption}}. This means it faithfully represents the behavior of all real components below the interface of the I/O library used in the implementation. These may include the I/O library, the operating system, the hardware,  the communication network, as well as potential attackers and link or node failures. 
	Recent work by \citet{ManskyHA20} demonstrates how to connect the verification of I/O behavior to a verified operating system to remove the I/O library and operating system from the trust assumptions. Their approach could be adapted to our setting. Other environment assumptions, such as the attacker model, remain and cannot be eliminated through formal proofs.


	\item \emph{Correct program configuration:} The programs are called with parameters \christoph{conforming to the configuration map $\vec{\gamma}$.}
	For instance, our case study assumes that each node program is initialized with parameters \christoph{$\vec{\gamma}(i) = (i, a)$} where $i$ is a node identifier and $a = \addr{\nextnode{i}}$ is the network address of $i$'s successor in the ring. 
	Verifying the configuration, which typically happens using scripts or manual procedures, is orthogonal to program correctness.

	\item \emph{Manual translation of I/O specification:} The I/O specification is translated correctly from the Isabelle/HOL to the code verifier tool's syntax. This translation is manual, but well-defined 
	and can thus be automated in the future by a translator implemented and verified in Isabelle.

	\item \emph{Toolchain soundness.} The verification logics and tools are sound and all proofs are thus correct. They agree on the semantics of I/O specifications and the code verifier satisfies the verifier assumption~\eqref{eq:verifier-assumption}. In our case, this concerns Isabelle/HOL and either VeriFast (which uses~Z3~\cite{de2008z3}) or Nagini (which depends on Viper~\cite{DBLP:conf/vmcai/0001SS16} and~Z3). 
	The trusted codebase could be reduced further by using a foundational verifier such as VST \cite{AppelVST12}, at the cost of a higher verification effort.
\end{enumerate}


\section{Case Studies: Fault-tolerance and Security}
\label{sec:case-studies}

We evaluate our methodology with two additional case studies that demonstrate the generality and versatility of Igloo. 
Concretely, we study a fault-tolerant, primary-backup replication system and an authentication protocol. 
These case studies showcase different features of our approach: (1) proofs of global, protocol-level properties, (2) environments with different types of networks as well as faulty and adversarial environments, (3) different component types with unbounded numbers of instances, and (4) \christoph{sequential and concurrent} implementations in different programming languages.
%
%

\subsection{Primary-backup Replication}
\label{sec:primary-backup-replication}
We apply our methodology to a primary-backup replication protocol \christoph{presented} by \citet[Sec. 2.3.1]{charron2010replication}. This case study exhibits the following features supported by our approach: (i) an environment that includes a fault model for components, (ii) reliable, in-order communication implemented by TCP, and (iii) sequential as well as concurrent implementations.

\subsubsection{Description}
\tk{
The primary-backup replication protocol maintains an append-only distributed data structure, called \emph{log}, which is a sequence of arbitrary data (\eg, database commands).
One server, the \emph{primary}, receives requests from clients to append elements to the log. 
The primary server first synchronizes a given append request with all other servers, the \emph{backups}, before 
extending its own log and responding to the client.
%
The protocol's goal is to maintain \emph{backup consistency}, \ie, the log stored on the primary when it replies to the client
is a prefix of the logs stored at all backups.
%
We assume a fail-stop fault model, where servers can fail \tk{but} not recover, and perfect failure detection, where all clients and servers eventually detect server failures (see, \eg,~\cite{DBLP:books/daglib/0025983}).
%
%
The servers are totally ordered, and initially, the first server is the primary. A backup server becomes the new primary once it detects that all previous servers in the order have failed. 
}
%
%

\subsubsection{Protocol Model}
\tk{
In this case study, we have chosen not to model an abstract version of the protocol (Step~1), but rather the concrete protocol (Step~2) directly. 
While the normal (fault-free) operation of the protocol is straightforward, the non-deterministic failures and their detection
add significant complexity.
When a new primary server takes over, its log may diverge from those of the backups. 
By synchronizing its log with those of the backup servers, it reestablishes a consistent state before responding to a client.
}
\tk{
Once backup $b$ has replied to a sync request from primary $a$, all logs contained in their states and sent between them are totally ordered in a prefix relation:
}
\[
\emph{ordered-wrt-prefix}(
\mklist{\textit{log}(a)} \cc 
\textit{transit}(b, a)  \cc 
\mklist{\textit{log}(b)}  \cc 
\textit{transit}(a, b) \cc
\mklist{\textit{pend}(a)}
).
\]
\tk{
Here, $\textit{pend}(a)$ is the primary's log including all pending additions, and $\textit{transit}(a, b)$ is the sequence of logs in transit from $a$ to $b$ (in sync requests or responses).
Additional inductive invariants and history variables are needed to derive this invariant. Careful modeling is also required for the failure detection.
The environment state contains a set \elive{}, consisting of all live servers. Since clients and servers may detect failures only after some delay, each of them has its own set \alive{} containing all servers except for those whose failure was noticed. As we show in an invariant, \alive{} sets are supersets of \elive{}.
In total, our proof of backup consistency relies on nine invariants. 
%
}

\subsubsection{Towards an Implementation}

\tk{
Our protocol model already includes the input and output buffers that are typically only added in Step 3.
This allows us to directly decompose the model into a trace-equivalent composition of client and server components and an environment (Step~4).  
}

\christoph{
Besides \emph{send} and \emph{receive}, the clients and servers have a third I/O event, \emph{detect-failure}, to query the failure detector. Its system part removes a server from the component's \alive{} set whereas its environment part has of a guard ensuring that the removed server is not in the \elive{}-set.
}

\tk{
We extract I/O specifications for both the server and the client component types (Step~5).
This extraction, as well as the equivalence proof between the component's event systems and their I/O specifications, follows the same standard pattern as in our security case study below.
Thus, this step could likely be automated in the future.
}


\subsubsection{Code Level}

We implement a sequential client and a concurrent server in Java. 
To handle requests in parallel, we split the server into multiple threads, communicating over shared buffers, \faw{guarded by a lock}. For TCP, we use Java's standard socket library.
\faw{
For failure detection, clients and servers get a failure detector object as an argument. 
This object provides methods to query whether a server has failed.
If Java's socket API determines that a connection attempt times out, the failure detector is queried. According to the protocol, failed servers are removed from the set of \alive{} servers, otherwise the last action is repeated.
For this case study, we provide a dummy failure detector instance, which stores the set of failed server ids.
When we kill a server in our execution script, the server process is terminated and its id is added to that  set.
The setup of contracts for trusted libraries and the verification of our client and server implementations in VeriFast against their respective I/O specifications closely follows our approach described in Section~\ref{ssec:component-implementation-and-verification}. 
}

\faw{
In the server, all shared data is protected by a lock.
For this lock, we define a \emph{monitor invariant}~\cite{LeinoMueller09},
declaring that the lock owns all shared data structures and 
the I/O permissions.
The ownership of these resources is transferred to a thread when it acquires the lock
and is transferred back when the lock is released.
The I/O permissions define which I/O operations may be performed \christoph{depending on the component model's state.} 
Since the implementation's behavior depends on the actual program state, in particular the state guarded by the lock, we also need to link the model state to the actual state in the monitor invariant.
We do this using an abstraction function.
%
\christoph{Thus, when executing an I/O operation, the associated model state update must be matched by a corresponding program state update before the lock can be released.
}}

\subsection{Two-party Entity Authentication}
\label{sec:security-protocols}

%
We also use our methodology to derive and implement a two-party authentication protocol standardized as ISO/IEC 9798-3. This case study demonstrates the following features of our approach: 
(i)~an adversarial environment, 
\christoph{
(ii)~the use of cryptography, 
}
(iii)~unreliable, unordered channels implemented by UDP, and 
(iv)~the use of data abstraction linking abstract message terms and their concrete \christoph{cryptographic} bitstring representations. 

\subsubsection{Description}

\christoph{
In standard (informal) Alice\&Bob notation, the protocol reads as follows.
\[
\begin{array}{lcll}
\textnormal{M1.} & A\rightarrow B & : & A,B,N_A \\
\textnormal{M2.} & B\rightarrow A & : & [N_B,N_A,A]_{\mathsf{pri}(B)}
\end{array}
\]
Here $N_A$ and $N_B$ are the nonces generated by the initiator $A$ and the responder $B$ respectively, and $[M]_{\mathsf{pri}(B)}$ denotes the digital signature of the message $M$ with $B$'s private key. 
First, the initiator generates a fresh nonce and sends it to the responder. Afterwards, the responder generates his own nonce, signs it together with the initiator's nonce and name, and sends the signed message to the initiator. The nonces provide replay protection. The protocol's authentication goal is that the initiator agrees with the responder on their names and the two nonces.
}

\subsubsection{Abstract and Protocol Models}

\christoph{
We follow the four-level refinement strategy for security protocols proposed by~\citet{DBLP:journals/jcs/SprengerB18}. Its levels corresponds to the first two steps of our methodology. We start from an abstract model of the desired security property, in our case, injective agreement~\cite{DBLP:conf/csfw/Lowe97a}. We then refine this into a protocol model that introduces the two protocol roles (\ie, Igloo component types) and their runs.
Each protocol run instantiates a role and stores the participants' names and received messages in its local state. We model an unbounded number of runs as a finite map from run identifiers to the runs' local states. 
The runs communicate over channels with intrinsic security properties. In our case, $A$  sends her nonce on an insecure channel to $B$, who returns it with his own nonce on an authentic channel back to~$A$. The attacker can read, but not modify, messages on authentic channels. 
}

\christoph{
In a second refinement, we represent messages as terms, use digital signatures to implement authentic channels, and refine the attacker into a standard Dolev-Yao attacker~\cite{DBLP:journals/tit/DolevY83} who completely controls the network.
The attacker's capabilities are defined by a closure operator $\id{DY}(M)$, denoting the set of messages he can derive from the set of observed messages $M$ using both message composition (such as encryption) and decomposition (such as decryption). All refined models  correspond to Igloo protocol models with different levels of detail. The refinement proofs imply that they all satisfy injective agreement. 
}

\subsubsection{Towards an Implementation}

\christoph{
We further refine the final protocol model into an interface model, where the messages are represented by bitstrings of an abstract type $\id{bs}$, which we later instantiate to actual bitstrings. We assume a surjective abstraction function $\alpha\!: \id{bs} \fun \id{msg}$ from bitstrings to messages. A special term $\kw{Junk}$ represents all unparsable bitstrings. 
We define a concrete attacker operating on (sets of) bitstrings by  
$\id{DY}_{bs} = \alpha^{-1} \circ \id{DY} \circ \alpha$, where 
$\alpha^{-1}$ is the inverse of $\alpha$ lifted to message sets. 
%
The simulation relation includes the equation $\alpha(\id{CIK})  = \id{IK}$, where $\id{CIK}$ and $\id{IK}$ respectively represent the concrete attacker's knowledge and the Dolev-Yao attacker's knowledge.
This allows us to transfer the Dolev-Yao model's security guarantees to the implementation.
}

\christoph{
We also augment each role's state with buffers for receiving and sending bitstring messages. The send and receive I/O events each take a network address and a bitstring message. 
%
The two roles' events still operate on message terms, but exchange messages with I/O buffers. For example, we refine a guard $m \in \id{IK}$ modeling a message reception in the protocol model by $bs \in \id{ibuf} \land \alpha(bs) = m$ in the interface model. 
%
The roles' events have several parameters, including the run id, the participants' names, and the long-term key bitstrings, which later become program parameters. 
}

\christoph{
Finally, we decompose the interface model into an environment and (an unbounded number of) initiator and  responder components. From these, we derive the initiator's and the responder's I/O specifications. Our framework provides lemmas to support the corresponding proofs.
}

\subsubsection{Code Level}

\christoph{
We implement each component type (protocol role) in both Java and Python. Each role's implementation sends and receives one message and checks that received messages have the correct form. 
%
Our implementation provides a (trusted) parsing and marshalling function for each type of abstract message (\eg, signatures, identifiers, pairs), each specified by a contract relating the message to its bitstring representation using $\alpha$. This yields a partial definition of $\alpha$, which we otherwise leave abstract to avoid modeling bit-level cryptographic operations.
Since $\alpha$ relates each bitstring to a unique message term (or $\kw{Junk}$), we ensure that every bitstring representing a non-junk message can be parsed unambiguously by tagging bitstring messages with their type. We employ widely-used cryptographic libraries: PyNaCl for Python and Java's standard library.
}%
\begin{figure}[t]
\begin{center}
\begin{lstlisting}[basicstyle=\scriptsize\sffamily,language=Python,caption={Simplified pseudocode implementation of a function for verifying signatures and its (trusted) specification. The pre- and postconditions relate the bitstrings to their abstract message representations that are used in the I/O specification. Variable \textit{a} is implicitly existentially quantified (denoted by a question mark).},captionpos=b,label={lst:signature_stub},escapechar=/]
def verify(signed_msg: bytes, key: bytes) -> bytes:
    #@ PRE: alpha(key) == Key(publicKey(?a)))
    #@ POST: alpha(signed_msg) == Crypt(Key(privateKey(a))), alpha(result)) 
    #@       && len(signed_msg) == len(result) + 1 + SIGNATURE_SIZE
    #@ EXCEPTIONAL POST: forall msg. alpha(signed_msg) != Crypt(Key(privateKey(a)), alpha(msg)) 
    if len(signed_msg) == 0 or signed_msg[0] != SIGNATURE_TAG:
        raise InvalidDataException("Invalid tag on signature") 
    return nacl.verify(signed_msg[1:])   # raises exception if signature is not valid
\end{lstlisting}
\end{center}
\vspace{-7pt}
\end{figure}

\christoph{
Listing~\ref{lst:signature_stub} shows the contract and implementation of the signature verification function. It checks the signature's tag and then calls the cryptographic library's signature verification function, which either returns the corresponding plaintext message or raises an exception for invalid signatures. 
The contract requires that the key bitstring represents some agent's public key and guarantees that the function terminates normally iff the input bitstring was signed with that agent's private key. 
}

\christoph{
We use VeriFast and Nagini to prove the implementations correct. 
As an overall result, we obtain a proof that the entire system satisfies the intended authentication property. 
}

%
%
%
%

\section{From Event Systems to I/O Specifications in Separation Logic}
\label{sec:theory}

\christoph{
A central aspect of our methodology is to soundly link protocol models and code. 
We use I/O specifications to bridge the substantial semantic gap between the component models that result from the interface model decomposition step and the code. 
The component models are given as event systems, whereas I/O specifications are separation logic formulas built over I/O permissions, possibly employing co-recursive predicates to specify non-terminating behavior.  What event systems and I/O specifications share is a trace semantics. We therefore define a translation of the components' event systems into I/O specifications and prove that they are trace-equivalent by establishing a simulation in each direction. It is this trace equivalence that, together with the verifier assumption~\eqref{eq:verifier-assumption} and the compositional refinement theorem (Corollary~\ref{cor:compositional-refinement}), enables the seamless switching from models to code verification in our methodology (cf.~Section~\ref{ssec:summary-and-guarantees}).%
}

Instead of translating component models directly into I/O specifications, we will pass through an intermediate translation to a sequential process calculus. 
This intermediate step has several benefits. 
First, it shifts the focus from data (guards and state updates) to I/O interactions.
Second, it introduces a minimal syntax for these interactions, providing a useful structure for the correctness proofs.
Third, it builds a bridge between two popular specification formalisms: process calculi on the modeling level and permission-based I/O specifications in separation logic on the code level.

The main challenge in proving our result stems from the disparate semantics of event systems and processes on one hand and I/O specifications on the other hand. 
Concretely, we must show that a process $P$ can simulate the traces induced by its corresponding  assertion $\phi_P$. As we shall see, an assertion's behavior is the intersection of all its models' behaviors. This is challenging as some models of $\phi_P$ exhibit spurious behavior not present in $P$ and \christoph{also} due to  the absence of a single model representing exactly the behavior of $\phi_P$.

\subsection{Background: Semantics of I/O Specifications for Separation Logic}
\label{ssec:io-separation-logic-semantics}

We slightly extend the semantics of the I/O specifications of~\citet{DBLP:conf/esop/Penninckx0P15} by enforcing a typing discipline on inputs by using a \emph{typing} function
$
\Ty : \Bio \times \Val \fun (\Pow(\Val)\setminus\{\emptyset\}),
$
which assigns to each I/O operation and output value a type, given as a non-empty set of accepted inputs. 
An I/O permission $\bio(t,v,w,t')$ and its input value $w$ are \emph{well-typed} if $w\in\Ty(\bio,v)$, and a chunk is well-typed if it is a well-typed I/O permission or a token. The typing function specifies a relational contract for each I/O operation. The set $\Ty(bio,v)$  typically captures the possible results of an I/O operation, which is useful to match I/O operations to I/O library functions. 

\subsubsection*{Assertion Semantics.}
\label{ssec:assertion-semantics}

The formal semantics of 
our assertions is co-inductively defined over \emph{heaps} $h \in \Heap$, where $\Heap = \multiset{\Chunk}$ is the set of multisets of chunks (see Section~\ref{ssec:io-separation-logic}), as follows.
\[
\begin{array}{lcl}
h \models b &  \Longleftrightarrow & b = \tru \\
h \models c   & \Longleftrightarrow & c \melem h \,\land\, \text{$c$ is well-typed}  \\
h \models \phi_1 \sand \phi_2  & \Longleftrightarrow & \exists h_1,h_2\in\Heap. \, h = h_1 \mplus h_2 \,\land\, h_1\models\phi_1 \,\land\, h_2\models\phi_2 \\
h \models \exists v.\,\phi  & \Longleftrightarrow & \exists v' \in \Val .\, h\models\phi[v'/v] \\
h \models \exists t.\,\phi  & \Longleftrightarrow & \exists t' \in \Place. \, h\models\phi[t'/t] 
\end{array}
\]

Note that a \emph{heap} here is different from \christoph{a program's heap memory}; its chunks represent permissions to perform I/O operations or tokens, not memory locations and their values.
Here, we elide the ordinary program heap for simplicity and since it is orthogonal to modeling I/O behavior.

The semantics of assertions satisfies the following monotonicity property.
\begin{lemma}[Monotonicity]
\label{lem:heap-extension}
If $h \models \phi$ then $g \mplus h \models \phi$. 
\end{lemma} 

\begin{example}
\label{ex:simple-IOspec-heaps}
Consider the I/O specification 
$
\phi = \token(t) \sand (\exists x, t', t''.\, \id{recv}(t, x, t') \sand \id{send}(t', 2x, t''))
$ 
from Example~\ref{ex:simple-IOspec}. 
Examples of heaps
that satisfy $\phi$ are
$h_1 = \mset{\token(t), \id{recv}(t,12,t_1), \id{send}(t_1,24,t_2)}$,
$h_2 = \mset{\token(t), \id{recv}(t,12,t), \id{send}(t,24,t)}$, and
$h_3 = h_1 \mplus \mset{\id{send}(t_1,35,t_2)}$.
More generally, all heaps satisfying $\phi$ have the form 
$
H_\phi(x, t',t'',h) = \mset{\token(t), \id{recv}(t,x,t'), \id{send}(t',2x,t'')} \mplus h
$ 
for some value $x$, places $t'$ and $t''$, and heap $h$. We will return to this example below.
\end{example}

\begin{figure}[t]
\begin{center}
\[
\begin{array}{c}
\infer[\BioRule]{\mset{\token(t),\bio(t,v,w,t')} \mplus h \hact{\bio(v,w)} \mset{\token(t')} \mplus h}{w\in\Ty(\bio,v)}  \\[1.5em] 
\infer[\ContradictRule]{\mset{\token(t),\bio(t,v,w,t')} \mplus h \hact{\bio(v,w')} \bot}{w \neq w' & w,w'\in\Ty(\bio,v)}  
\qquad
\infer[\ChaosRule]{\bot \hact{\bio(v,w)} \bot}{w \in \Ty(bio,v)}  
\end{array}
\]
\vspace{-4mm}
\caption{Heap transition rules.}
\label{fig:heap-trans}
\end{center}
\vspace{-2mm}
\end{figure}

\subsubsection*{Heap Transitions.}
\label{ssec:heap-transitions}

Heaps have a transition semantics, where I/O permissions are consumed by pushing a token through them. This semantics is given by the event system $\ESH = (\option{\Heap}, \ActBio, \;\hact{})$ with the set of states $\option{\Heap}$ and the set of events
$
\ActBio = \{\bio(v,w) \mid \bio\in\Bio \land v,w\in\Val\},
$
called \emph{I/O actions}. 
Note that $\bio$ is overloaded, with the 2-argument version yielding trace events and the 4-argument one defining a chunk. 
An I/O action $\bio(v,w)$ is \emph{well-typed} if $w\in\Ty(\bio,v)$ and a trace $\tau \in \ActBio^*$ is well-typed if all its events are.

The transition relation $\hact{}$ of $\ESH$ is defined by the rules in Figure~\ref{fig:heap-trans} and mostly matches the place-I/O-permission multigraph intuition given in Section~\ref{ssec:io-separation-logic}.
The rule~\BioRule\ corresponds to a normal heap transition executing an I/O operation. The input read is well-typed. The token moves to the I/O permission's target place and the permission is consumed and removed from the heap. 
The rule~\ContradictRule\ describes the situation where a transition $\bio(v,w)$ would be possible, but the environment provides an
input $w' \neq w$ that is different from the one predicted by the chunk. In this case, the heap can perform a transition $\bio(v,w')$ to the special state $\bot$.
In this state, arbitrary (well-typed) behavior is possible by the rule~\ChaosRule.  
Hence, all traces of $\ESH$ are well-typed.
For a set of heaps $H$, we define the set of traces of $H$ to contain the traces executable in \emph{all} heaps of $H$, \ie,
\[
\tracesH(H) = \{\tau \mid \forall h \in H.\, \tau \in \traces(\ESH, h)\}.
\]
The set of traces of an assertion $\phi$ is then defined to be the set of traces of its heap models, \ie, 
\[
\tracesH(\phi) = \tracesH(\{h \mid h \models \phi\}).
\]
%
This universal quantification over all heap models of an assertion constitutes the main challenge in our soundness proof (Theorem~\ref{thm:gorilla-glue}).
Let us now look at an example illustrating these definitions.
\begin{example}[Heap and assertion traces]
\label{ex:simple-IOspec-traces}
Consider the heap models $h_1$, $h_2$, and $h_3$ of the I/O specification~$\phi$ from Example~\ref{ex:simple-IOspec-heaps}. 
First focusing on \emph{regular \christoph{behaviors}}, \ie, ignoring the rules~\ContradictRule\ and~\ChaosRule, their traces are given by the following sets, where $\,\pc$ denotes prefix closure: 
\begin{itemize}
\item $\traces(\ESH, h_1) = \{ \id{recv}(12) \cc \id{send}(24) \}\pc $,
\item $\traces(\ESH, h_2) = \traces(\ESH, h_1) \cup \{ \id{send}(24) \cc \id{recv}(12) \}\pc$, and
\item $\traces(\ESH, h_3) = \traces(\ESH, h_1) \cup \{ \id{recv}(12) \cc \id{send}(35) \}\pc$.
\end{itemize}
%
The first heap, $h_1$, exhibits an instance of the expected behavior: receive a value and send the doubled value. The heaps $h_2$ and $h_3$, however, also allow unintended behaviors. Heap $h_2$ has a trace  where receive and send are inverted. This comes from the semantics of existential quantification, which does not ensure that the places are distinct. Heap $h_3$ can send a value different from the doubled input value, which is possible due to the monotonicity property in Lemma~\ref{lem:heap-extension}. 
Due to these additional behaviors, which we call \emph{spurious}, the set $\tracesH(\psi)$ of traces of an I/O specification~$\psi$ is defined to contain those traces that are possible in \emph{all} heap models of $\psi$. The three heaps above only share the traces of $h_1$, which corresponds to the intended behavior.

Note that these spurious behaviors are not an artifact of the particular formalism we use, but a standard part of the permission-based specification style of separation logics in general. For example, all program heaps satisfying a standard points-to assertion $\pointsto{x}{e}$ allow the program to dereference the pointer $x$, but some heaps may also allow dereferencing the pointer $z$ because~$z$ and~$x$ happen to alias in a particular interpretation (analogous to ``aliasing'' places in $h_2$), or, for logics with  monotonicity, may contain (and therefore allow access to) extra memory pointed to by~$y$.  However, like in our case, the program logic must not allow dereferencing $y$ or $z$ because it is not possible in \emph{all} program heaps satisfying the assertion.

The rules~\ContradictRule\ and~\ChaosRule\ add, for any \christoph{regular} trace of the form $\tau_1 \cc \id{recv}(w) \cc \tau_2$, \christoph{a spurious traces} of the form $\tau_1 \cc \id{recv}(w') \cc \tau$ for \christoph{each} well-typed $w' \neq w$ and well-typed trace $\tau$. These rules formalize that a heap \christoph{reading} some (well-typed) input different from the one predicted by the I/O permission may behave arbitrarily. 
\christoph{For example, both $h_1' = \mset{\token(t), \id{recv}(t,19,t_1), \id{send}(t_1,38,t_2)}$ and $h_1$ are models of $\phi$ and $\nil$ is their only shared regular trace. However, the regular traces of $h_1'$ are also spurious traces of $h_1$ and vice versa. Hence, $\tracesH(\{h_1,h_1'\})$ consists of the regular traces of $h_1$ and $h_1'$. 
This ensures that
$
\tracesH(\phi) = \{  \id{recv}(x) \cc \id{send}(2x)  \mid x \in \Val \}\!\pc
$
is the trace property intended by the assertion $\phi$.
}
Without these two rules, we would have $\tracesH(\phi) = \{ \nil \}$. 
%
\end{example}

\subsection{Embedding I/O-guarded Event Systems into Processes}


We co-inductively define a simple language of (sequential) processes:
\[
  P ::=_{\corec} \Null \mid \bio(v, z).P \mid P_1 \choice P_2.
\]
Here, $\Null$ is the inactive process, $\bio(v, z).P$ is prefixing with an I/O operation, which binds the input variable $z$ in $P$, and $P_1 \choice P_2 $ is a binary choice operator. Let $\Proc$ be the set of all processes.

We can then co-recursively define processes. For example, we define a countable choice operator $\bigchoice_{v\in S} P(v)$ over a set of values $S$ with $\Null$ as the neutral element, analogous to the definition of the iterated separating conjunction. We can also co-recursively define non-terminating processes.

\begin{example}
A process corresponding to the I/O specification from Example~\ref{ex:IOspec-read-write} is specified by $\id{RSP}(0)$, where
$
\id{RSP}(a) =_\corec \id{recv}(z). \,\ifte{z>0}{\id{send}(a+z).\id{RSP}(a+z)}{\Null}.
$
\end{example}

The operational semantics of processes is given by the event system $\ESP = (\Proc, \ActBio, \opsem{})$, where the transition relation $\opsem{}$ is inductively defined by the following rules:
\[
\begin{array}{c}
\infer[\textsf{Pref}]{\bio(v,z).P \opsem{\bio(v,w)} P[w/z] }{w \in \Ty(\bio, v)} \qquad
\infer[\textsf{Choice}_1]{P_1 \choice P_2 \opsem{a} P_1'}{P_1 \opsem{a} P_1'} \qquad
\infer[\textsf{Choice}_2]{P_1 \choice P_2 \opsem{a} P_2'}{P_2 \opsem{a} P_2'}
\end{array}
\]
We write $\traces(P)$ as a shorthand for $\traces(\ESP, P)$. 


\subsubsection*{Translation}
We define a translation from I/O-guarded event systems $\GES = (S, \ActBio, G, U)$ to processes. The process $\proc(\GES, s)$ represents $\GES$ in state $s$ and is co-recursively defined by 
\begin{align*}
\proc(\GES, s) =_{\corec} & \bigchoice_{\bio\in\Bio}\, \bigchoice_{v\in\Val}  
\ifte{G_{(\bio,v)}(s)}{\bio(v, z).\, \proc(\GES, U_{\bio(v, z)}(s))}{\Null}. 
\end{align*}
Recall that here we borrow the conditional from our meta-language HOL\@.
The following correctness result is established by a simulation in each direction.

\begin{theorem}[Correctness of event system translation]
\label{thm:ioges-into-proc-correct}
For any I/O-guarded event system $\GES=(S,\ActBio,G,U)$ and state $s\in S$, we have $\traces(\GES,s) = \traces(\proc(\GES, s))$. 
\end{theorem}
%

\subsection{Embedding Processes into I/O Specifications}
\label{ssec:processes-into-iospecs}


We now co-recursively define the embedding $\emb$ from processes and places into I/O specifications:
\begin{align*}
 \emb(\Null, t) & =_\corec \tru \\
 \emb(\bio(v, z).P, t) & =_\corec \exists t',z'.\, \bio(t, v, z', t') \sand \emb(P[z'/z], t') \\
 \emb(P_1 \choice P_2, t) & =_\corec  \emb(P_1,t) \sand \emb(P_2, t).
\end{align*}
We define the \emph{process assertion} of $P$ by $\emb(P) = \exists t.\, \token(t) \sand \emb(P,t)$.
We then prove by co-induction that countable choice translates to iterated separating conjunction.
\begin{lemma}
\label{lem:bigchoice}
 $\emb(\bigchoice_{v\in S} P(v), t) = \sall v\in S.\, \emb(P(v), t)$.
\end{lemma}

We now turn to our main result, namely, the trace equivalence of process $P$ and its I/O specification $\emb(P)$. We focus on the intuition here and defer the formal details to~\fullversionref{app:theory-details}.

\begin{theorem}[Correctness of process translation] 
\label{thm:gorilla-glue}
$\traces(P) =  \tracesH(\emb(P))$.
\end{theorem}

The proof follows from Propositions~\ref{prop:tracesP-subset-traces-embP}, \ref{prop:traces-embP-subset-traces-canP}, and~\ref{prop:traces-canP-subset-tracesP} to which the remainder of Section~\ref{sec:theory} is devoted.
Together with Theorem~\ref{thm:ioges-into-proc-correct}, this result allows us to translate any I/O-guarded event system $\ES$ modeling some component of our system into an I/O specification $\phi_{\ES} = \emb(\proc(\ES))$
with identical behavior. We can then use $\phi_{\ES}$ as a specification for the code implementing $\ES$'s behavior. 


The left-to-right trace inclusion of this theorem is captured by the following proposition, \christoph{which we prove by} a simulation between process $P$ and heap models of $\emb(P)$ (see~\fullversionref{sapp:processes-into-iospecs}).
\begin{proposition}
\label{prop:tracesP-subset-traces-embP}
$\traces(P)  \subseteq  \tracesH(\emb(P))$.
\end{proposition}


The main difficulty lies in the proof of the reverse set inclusion and stems from the meaning of $\tracesH(\emb(P))$, which contains exactly those traces $\tau$ that are a trace of \emph{all} models of $\emb(P)$. From Example~\ref{ex:simple-IOspec-traces}, we know that many models of $\emb(P)$ (or of any  assertion $\phi$ for that matter) exhibit spurious behaviors that are not in $\tracesH(\emb(P))$ (or in $\tracesH(\phi)$, respectively). 
Therefore, picking an arbitrary heap model of $\emb(P)$ and trying to simulate its transitions by the process $P$'s transitions will fail. Instead, we restrict our attention to \emph{canonical} models that do not exhibit spurious behaviors.
We denote by $\can(P)$ the set of all canonical models of $P$ (introduced in Section~\ref{ssec:canonical-models}).
We then decompose the proof into the following chain of set inclusions:
\begin{equation}
\label{eq:trace-inclusions}
\tracesH(\emb(P)) \subseteq \tracesH(\can(P)) \subseteq \traces(P).
\end{equation}
The first inclusion expresses that the canonical models cover all behaviors of $\emb(P)$. We will establish the second inclusion by simulating the behavior of canonical models by process transitions.

\subsection{Canonical Heap Models for Processes}
\label{ssec:canonical-models}

A natural canonical model candidate for a process $P$ would be the heap $h_P$ that is isomorphic to $P$'s computation tree, where a process $\bio(v,w).Q$ would result in one I/O permission $\bio(t, v, w, t_w)$ for each input $w$ on the heap. Although this proposal avoids spurious behaviors due to additional permissions and place identifications (\cf~Example~\ref{ex:simple-IOspec-traces}), it fails as the following example shows.

\begin{example}[Failed attempt]
Let $P = in(x).out(x).\Null$, $\Val = \Bool$, and $\Place = \List{\{\Left,\Right\}}$ (for tree positions). Then $h_P$ contains both I/O permissions $in(\nil, \fal, \Left)$ and $in(\nil, \tru, \Right)$.  
This would lead to $\tracesH(\can(P)) = \traces(h_P) = \{\nil\} \cup \{in(v) \cc \tau \mid v \in \Bool, \tau \in \List{\ActBio} \}$ according to the rules~\ContradictRule\ and~\ChaosRule\ and hence to $\tracesH(\can(P)) \supset \traces(P)$.
\end{example}

We will therefore construct the canonical heap models of a process $P$ with respect to an input schedule, which \christoph{is essentially a prophecy variable that} uniquely determines the inputs read by the process. An \emph{input schedule} is a function $\rho : \ActBio^* \times \Bio \times \Val \fun \Val$ mapping an I/O trace $\tau$, an I/O operation $\bio$, and an output value~$v$ to an input value $\rho(\tau,\bio,v)$. Hence, there will be a canonical model $\cmodel(P,\rho)$ for each input schedule $\rho$, which intuitively corresponds to the projection of $P$'s computation tree to the inputs prescribed by $\rho$. The set $\can(P)$ contains such a model for each input schedule $\rho$. Our construction uses the set of places $\Place = \List{\{\Left,\Right\}}$, \ie, the places are positions of a binary tree. The inputs being fixed, the only branching stems from the choice operator. The following example illustrates our construction. We defer its formal definition \christoph{and the proofs of the corresponding results to~\fullversionref{sapp:canonical-models}.}


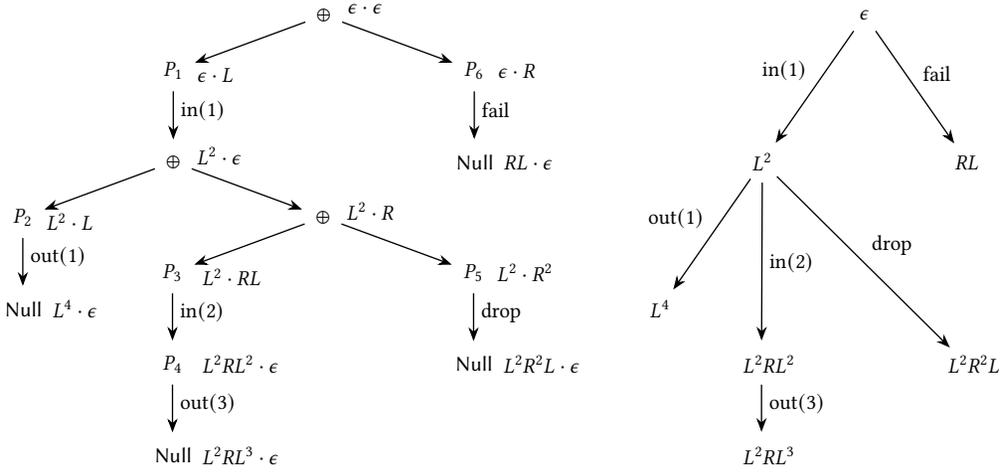
\begin{figure*}[t]
    \centering
\vspace{-7mm}
\begin{tikzpicture}[shorten >=1pt,node distance=1.5cm,on grid,auto,initial text={},
      >=Stealth, auto, scale = 0.8, transform shape, state/.style={circle,inner sep=0pt}
] 
   \node[state,processstyle] (e) at (0,0) {$\choice$};
   \node[state,processstyle] (L) [position=-160:{2cm} from e] {$P_1$};
   \node[state,processstyle] (L2) [below=of L] {$\choice$};
   \node[state,processstyle] (R) [position=-20:{2cm} from e] {$P_6$};
   \node[state,processstyle] (RL) [below=of R] {$\Null$};
   \node[state,processstyle] (L2R) [position=-20:{2cm} from L2] {$\choice$};
   \node[state,processstyle] (L2RL) [position=-160:{2cm} from L2R] {$P_3$};
   \node[state,processstyle] (L2RL2) [below=of L2RL] {$P_4$};
   \node[state,processstyle] (L3) [position=-160:{2cm} from L2] {$P_2$};
   \node[state,processstyle] (L4) [below=of L3] {$\Null$};
   \node[state,processstyle] (L2R2) [position=-20:{2cm} from L2R] {$P_5$};
   \node[state,processstyle] (L2R2L) [below=of L2R2] {$\Null$};
   \node[state,processstyle] (L2RL3) [below=of L2RL2] {$\Null$};

\draw[->,line width=0.5pt]
    (e) edge node {} (L)
    (L) edge node [yshift=3pt] {in$(1)$} (L2)
    (e) edge node {} (R)
    (R) edge node [yshift=3pt] {fail} (RL)
    (L2) edge node {} (L3)
    (L2) edge node {} (L2R)
    (L2R) edge node {} (L2RL)
    (L2R) edge node {} (L2R2)
    (L3) edge node [yshift=3pt] {out$(1)$}  (L4)
    (L2RL) edge node [yshift=3pt] {in$(2)$} (L2RL2)
    (L2R2) edge node [yshift=3pt] {drop} (L2R2L)
    (L2RL2) edge node [yshift=3pt] {out$(3)$} (L2RL3)
;

   \node[state,ppos] (!e) [right=\normalDist of e,yshift=3pt] {$\POSee{}$};
   \node[state,ppos] (!L) [right=\normalDist of L,yshift=-2pt] {$\POSel{L}$};
   \node[state,ppos] (!L2) [right=\normalDist of L2,yshift=3pt] {$\POSle{L^2}{}$};
   \node[state,ppos] (!R) [right=\normalDist of R] {$\POSel{R}$};
   \node[state,ppos] (!RL) [right=\bigDist of RL] {$\POSle{RL}$};
   \node[state,ppos] (!L2R) [right=\normalDist of L2R,yshift=3pt] {$\POS{L^2}{R}$};
   \node[state,ppos] (!L2RL) [right=\bigDist of L2RL,yshift=-2pt] {$\POS{L^2}{RL}$};
   \node[state,ppos] (!L2RL2) [right=\bigDist of L2RL2] {$\POSle{L^2RL^2}$};
   \node[state,ppos] (!L2RL3) [right=\bigDist of L2RL3]  {$\POSle{L^2RL^3}$};
   \node[state,ppos] (!L3) [right=\normalDist of L3,yshift=-2pt] {$\POS{L^2}{L}$};
   \node[state,ppos] (!L4) [right=\bigDist of L4] {$\POSle{L^4}$};
   \node[state,ppos] (!L2R2) [right=\normalDist of L2R2] {$\POS{L^2}{R^2}$};
   \node[state,ppos] (!L2R2L) [right=\bigDist of L2R2L] {$\POSle{L^2R^2L}$};

   \node[state,processstyle] (pie) [right=\piDist of e] {$\nil$};
   \node[state,processstyle] (piL2) [right=\piDist of L2,xshift=8mm] {$L^2$};
   \node[state,processstyle] (piRL) [right=\piDist of RL,xshift=-8mm] {$RL$};
   \node[state,processstyle] (piL2RL2) [right=\piDist of L2RL2,xshift=8mm] {$L^2RL^2$};
   \node[state,processstyle] (piL4) [right=\piDist of L4,xshift=16mm] {$L^4$};
   \node[state,processstyle] (piL2R2L) [right=\piDist of L2R2L,xshift=-8mm] {$L^2R^2L$};
   \node[state,processstyle] (piL2RL3) [right=\piDist of L2RL3,xshift=8mm] {$L^2RL^3$};

\draw[->,line width=0.5pt]
    (pie) edge node [above left] {in$(1)$} (piL2)
    (pie) edge node {fail} (piRL)
    (piL2) edge node [above left] {out$(1)$} (piL4)
    (piL2) edge node {in$(2)$} (piL2RL2)
    (piL2) edge node {drop} (piL2R2L)
    (piL2RL2) edge node [yshift=3pt]{out$(3)$} (piL2RL3)
;
\end{tikzpicture}
\vspace{-6mm}
\caption{Process $P$ and the schedule~$\rho$ of Example~\ref{ex:canonical-model} (left) and resulting canonical model (right). 
Each process (node in the left graph) is annotated with its position $cpos = ppos \cc x$.
%
}
\label{fig:canonical-model}
\end{figure*}

\begin{example}[Canonical model]
\label{ex:canonical-model}
Consider the process $P$ defined by 
\begin{align*}
     P     & = \id{in}(x).Q(x) \choice \id{fail}.\Null   & 
     Q(x) & = \id{out}(x).\Null 
                      \choice (\id{in}(y).\id{out}(x + y).\Null
                      \choice \id{drop}.\Null).
\end{align*}
For simplicity, the I/O operations $\id{drop}$ and $\id{fail}$ have no arguments. Let $\rho$ be the input schedule defined by $\rho(\tau, bio, v) = \len{\tau} + 1$. Figure~\ref{fig:canonical-model} (left) shows the projection of $P$'s syntax tree to the input schedule $\rho$. Edges arising from action prefixes are labeled with the corresponding action. Each node is annotated with its current position $cpos = ppos \cc x$, which is composed of $ppos$, the target position of the previous action-labeled edge in the tree (or $\nil$ if there is none), and a rest~$x$. 
Each edge labeled by some action $\bio(v, w)$ and connecting position $cpos = ppos \cc x$ to $cpos \cc \Left$ translates into an I/O permission $\bio(ppos, v, w, cpos\cc\Left)$ in the resulting canonical heap $\cmodel(P, \rho)$, which is shown in Figure~\ref{fig:canonical-model} (right).
\end{example}

The following result states that the canonical model for a process $P$ and a schedule $\rho$ is indeed a model of the assertion corresponding to the process $P$. The first inclusion in~\eqref{eq:trace-inclusions} \christoph{then} easily follows. 
\begin{proposition}[Canonical model property]
\label{prop:canonical-model}
$\cmodel(P, \rho) \models \emb(P, \nil)$ for all processes $P$ and well-typed schedules $\rho$.
\end{proposition}

\begin{proposition}
\label{prop:traces-embP-subset-traces-canP}
$\tracesH(\emb(P)) \subseteq \tracesH(\can(P))$.
\end{proposition}
%

\subsection{Processes Simulate Canonical Models}
\label{ssec:processes-simulate-canonical-models}


We now \christoph{turn to} the second trace inclusion in~\eqref{eq:trace-inclusions}: each trace of the canonical model set $\can(P)$ is also a trace of $P$.
Writing $\addtoken{\cmodel}(P,\rho)$ for the canonical model $\cmodel(P,\rho)$ with a token added at its root place, we would like transitions of the heap $\addtoken{\cmodel}(P,\rho)$ to lead to a heap $\addtoken{\cmodel}(P',\rho)$ for some process $P'$, so we can simulate it with the corresponding process transition from $P$ to $P'$. 

\christoph{There are two obstacles to this plan:} (1) dead heap \christoph{parts, which} correspond to untaken choices in processes $P \choice Q$ and cannot perform any transitions, and (2) chaotic \christoph{transitions where,} given a trace \christoph{of the set of canonical models $\can(P)$, 
some of the models $\addtoken{\cmodel}(P,\rho)$ in $\can(P)$} transit to the ``chaotic'' state $\bot$ at some point along the trace. The problem here is that a given process cannot in general simulate the (arbitrary) I/O actions \christoph{that are} possible in \christoph{the state $\bot$}. 

Our proofs must take such dead heap parts into account to address problem (1) and carefully pick a particular schedule to avoid problem (2). Here, we focus on problem~(2) from an intuitive perspective (see~\fullversionref{sapp:processes-simulate-canonical-models} for a more precise and detailed account). Its solution is based on the observation that executing some I/O action $\bio(v, w_{\rho})$ with \emph{scheduled input $w_{\rho}=\rho(\tau, \bio, v)$} from $\addtoken{\cmodel}(P,\rho)$ indeed leads to a heap $\addtoken{\cmodel}(P',\rho)$ for some process $P'$ (and, in particular, not to~$\bot$).
Hence, to simulate a given trace $\tau$ of the heap $\addtoken{\cmodel}(P, \rho)$ by transitions of the process $P$, we must ensure that the schedule $\rho$ is consistent with the trace $\tau$. We therefore define a witness schedule $\rhowit(\tau)$, which returns the inputs appearing on the trace $\tau$ and has the property:
\begin{equation}
\label{eq:can-mod-rhowit-trace}
\addtoken{\cmodel}(P,\rhowit(\tau))\trans{\tau} h = \addtoken{\cmodel}(P',\rhowit(\tau))  
\end{equation}
for some process $P'$, \ie, in particular, $h \neq \bot$. 
The final trace inclusion in Equation~\eqref{eq:trace-inclusions} then follows immediately, since any trace $\tau \in \tracesH(\can(P))$ is also a trace of $\addtoken{\cmodel}(P,\rhowit(\tau))$.
\begin{proposition}
\label{prop:traces-canP-subset-tracesP}
$\tracesH(\can(P)) \subseteq \traces(P)$.
\end{proposition}


\section{Related Work}
\label{sec:related-work}

Numerous formalisms have been developed for modeling and verifying systems. In the following, we focus on those approaches that combine models and code, and target distributed systems.

\subsubsection*{Model Verification with Code Extraction}

Various approaches verify models of distributed systems in formalisms that support the extraction of executable code. The following four approaches are all embedded in Coq and support the extraction of OCaml programs.

In Verdi~\cite{DBLP:conf/pldi/WilcoxWPTWEA15,DBLP:conf/cpp/WoosWATEA16}, a system is specified by defining types and handlers for external I/O and for network messages. The developer can focus on the application and its correctness proof by essentially assuming a failure-free environment. These assumptions can be relaxed by applying Verdi's verified system transformers to make the application robust with respect to communication failures or node crash failures. 
DISEL~\cite{DBLP:journals/pacmpl/SergeyWT18} offers a domain-specific language for defining protocols in terms of their invariants and atomic I/O primitives. It enables the modular verification of programs that participate in different protocols, using separation logic to represent protocol state separation. Component programs are verified in the context of one or more protocol models using a Hoare logic embedded in a dependent type theory. 
The program verification can be understood as a single refinement step.  
Velisarios~\cite{DBLP:conf/esop/RahliVVV18} is a framework for verifying Byzantine fault-tolerant state-machine replication protocols in Coq based on a logic of events. It models systems as deterministic state machines and provides an infrastructure for modeling and reasoning about distributed knowledge and quorum systems.  
%
Chapar~\cite{DBLP:conf/popl/LesaniBC16} is a formal framework in Coq for the verification of causal consistency for replicated key-value stores. The technique uses an abstract operational semantics that defines all the causally-consistent executions of a client of the store. The implementation of the store is verified by proving that its concrete operational semantics refines this abstract semantics. 
%
%
\citet{nfm20-liu} model distributed systems in Maude's rewriting logic~\cite{DBLP:conf/maude/2007}. These are compiled into distributed implementations using mediator objects for the TCP communication. They prove that the generated implementation is stuttering equivalent to the original model, hence preserving next-free CTL* properties. The implementation runs in distributed Maude sessions.

All of these approaches enable the development of distributed systems that are correct by construction. However, code extraction has three major drawbacks. 
First, the produced code is either purely functional or based on rewriting logic, which precludes common optimizations (\eg, using mutable heap data structures). 
Second, it is difficult for extracted code to interface existing software modules such as libraries; incorporating existing (possibly unverified) modules is often necessary in practice. 
Third, the approaches prescribe a fixed implementation language; however, it is often useful in practice to be able to combine components, such as clients and servers, written in different languages.
Our approach avoids all three problems by supporting the bottom-up development and verification of efficient, flexible implementations.

PSync~\cite{DBLP:conf/popl/DragoiHZ16} is a domain-specific language for implementing \christoph{round-based} distributed, fault-tolerant systems. PSync programs are executed via an embedding into Scala. A dedicated verifier allows one to prove safety and liveness properties of PSync programs, and a refinement result shows that these carry over to the executable system. \christoph{The focus of PSync is mostly on developing specific verified distributed \emph{algorithms} rather than entire software systems.} 


\subsubsection*{Combinations of Model and Code Verification}

The works most closely related to ours are those of~\citet{DBLP:conf/cpp/Koh0LXBHMPZ19} and of~\citet{DBLP:conf/ifm/OortwijnH19}. The former work
is part of DeepSpec~\cite{DBLP:conf/oopsla/Pierce16}, which is a research program with the goal of developing fully-verified software and hardware.
The DeepSpec developments are based on the Verified Software Toolchain~(VST)~\cite{CaoBGDA18}, a framework for verifying C programs via a separation logic embedded in Coq. 
\citet{DBLP:conf/cpp/Koh0LXBHMPZ19} use \emph{interaction trees}~\cite{DBLP:journals/pacmpl/XiaZHHMPZ20}, which are similar to our processes, to specify a program's I/O behavior and directly embed these into VST's separation logic using a special predicate. 
In contrast, our embedding of processes into separation logic using the encoding of~\citet{DBLP:conf/esop/Penninckx0P15} 
allows us to apply standard separation logic and existing program verifiers.
In both their and our work, a successful program verification guarantees an inclusion of the program's I/O traces in those of the I/O specification or interaction tree. 
\citet{DBLP:conf/cpp/Koh0LXBHMPZ19} verify a simple networked server in a TCP networking environment, for which they use two interaction trees at different abstraction levels and relate them by a form of contextual refinement that establishes linearizability. 
Their paper leaves open the question whether their approach can be used to verify system-wide global properties of distributed systems with different types of components and operating in different environments (\eg, exhibiting faulty and adversarial behavior). For example, it is unclear whether they could verify our case study protocols.
%
\citet{DBLP:conf/ifm/OortwijnH19} use a process calculus for modeling, which they embed into a concurrent separation logic~(CSL). Their approach relies on automated tools and combines the mCRL2 model checker with an encoding of CSL into Viper. The modeling-level expressiveness is limited by mCRL2 being a finite-state model checker. Moreover, while the soundness of CSL implies the preservation of state assertions from modeling to implementation level, it is unclear whether arbitrary trace properties are preserved.

%



IronFleet~\cite{DBLP:conf/sosp/HawblitzelHKLPR15} combines TLA-style refinement with code verification. Abstract models as well as the implementation are expressed in Dafny~\cite{Leino10}. Dafny is a powerful verification framework that supports, among other features, mutable heap data structures, inductive and coinductive data types, and proof authoring. 
Reasoning is supported by an SMT solver, which \mar{is restricted} to first-order logic.
\mar{Dafny enables different kinds of higher-order reasoning by encoding it into first-order logic internally, but nevertheless has some restrictions both in expressivity and practicality for larger proofs when compared to native higher-order theorem provers.}
By using Isabelle/HOL as modeling language, our approach provides the full expressiveness of higher-order logic, which also allows us to formalize our meta-theory. 
By using a single framework, Ironfleet avoids the problems we had to solve when linking abstract models to separation logic specifications. However, it lacks the flexibility to support different logics or modeling languages. Dafny currently compiles to sequential C\#, Go, and JavaScript, while existing separation logic based verifiers support concurrent implementations and allow developers to write the code directly in familiar programming languages rather than in Dafny.
IronFleet supports both safety and liveness properties, whereas our approach focuses on safety properties and leaves liveness as future work. 


Project Everest~\cite{BhargavanBDFHHI17} uses an approach similar to IronFleet to develop a verified implementation of TLS\@. An abstract implementation is developed and verified in Low$^\ast$~\cite{ProtzenkoZRRWBD17}, a subset of F$^\ast$~\cite{SwamyHKRDFBFSKZ16} geared toward imperative C-like code that is compiled to C\@. A main focus of this project is on verifying cryptographic algorithms. Like 
IronFleet, Low$^\ast$ verification uses an SMT solver and the extracted C code is sequential.

\section{Conclusions and Future Work}
\label{sec:conclusions}

We proposed a novel approach for the formal development of distributed systems. 
Our approach combines the top-down development of system models via compositional refinement with bottom-up program verification. This supports a \christoph{clean} separation of concerns and simplifies the individual verification tasks, which is crucial for managing the additional complexity arising in systems operating in faulty or adversarial environments. For program verification, we support state-of-the-art separation logics, which support mutable heap data structures, concurrency, and other features needed to develop efficient, maintainable code. We demonstrated that our approach bridges the gap between abstract models and concrete code, both through the theoretical foundations underpinning its soundness and \christoph{with} three complete case studies.  
The theory and case studies are mechanized in Isabelle/HOL 
and the Nagini and VeriFast program verifiers.

As future work, we plan to reduce the need for boilerplate Isabelle code by automating the translation of interface models into the components' I/O specifications that are input to the code verifiers.
%
We also plan to support liveness properties, which will require a more complex refinement framework in the style of TLA~\cite{Lamport94}, including support for fairness notions.
%
Finally, we are currently applying our approach to verify substantial parts of the SCION secure Internet architecture~\cite{DBLP:series/isc/PerrigSRC17}. We show protocol-level global security properties in the Dolev-Yao symbolic attacker model and verify the I/O behavior (as well as memory safety, secure information flow, and other properties) of the currently deployed implementation of  SCION routers. 
\bibliography{igloo-paper}

\newpage
\appendix

%
%
%
%

\iffullversion

\section{Theory Details}
\label{app:theory-details}

This section provides details on the main part of our soundness theorem, the equivalence of processes and their I/O specifications (Theorem~\ref{thm:gorilla-glue}).

\[\traces(P) =  \tracesH(\emb(P)).\]

Recall that we prove this theorem using a series of trace inclusions (Propositions~\ref{prop:tracesP-subset-traces-embP}, \ref{prop:traces-embP-subset-traces-canP}, and~\ref{prop:traces-canP-subset-tracesP}).

\[\traces(P) \subseteq \tracesH(\emb(P)) \subseteq \tracesH(\can(P)) \subseteq \traces(P).\]

\subsection{Formal Definitions and Proofs for Section~\ref{ssec:processes-into-iospecs}}
\label{sapp:processes-into-iospecs}

\subsubsection{Process Traces Are Process Assertion Traces}

\begin{proposition}
$\traces(P)  \subseteq  \tracesH(\emb(P))$.
\end{proposition}

\begin{proof}
It suffices to prove that $\traces(P) \subseteq \traces(ho)$ for any $P$ and $ho$ in the simulation relation 
\[
R(P,ho) = (\exists t, h.\, ho = \mset{\token(t)} \mplus h \land h \models \emb(P, t)) \lor ho = \bot.
\]
The proof proceeds by establishing a simulation between $P$ and $ho$ using this relation. 
If $P$ is related to a heap $h$ (first disjunct in $R$) then a given transition of $P$ can either be simulated by a transition to another heap $h'$ (using rule \BioRule) or to $\bot$ (using rule \ContradictRule). In each case, the resulting states are again related by $R$. We prove this by induction on the operational semantics of processes.
%
Otherwise, if $P$ is related to $\bot$, then any of $P$'s transitions can be simulated using rule \ChaosRule.
\end{proof}

\subsection{Formal Definitions and Proofs for Section~\ref{ssec:canonical-models}}
\label{sapp:canonical-models}

\subsubsection{Formal Definition of Canonical Model $\cmodel(P, \rho)$}

Recall from Example~\ref{ex:simple-IOspec-traces} that there are two sources of spurious behaviors: (a) unintended control flow stemming from the identification of places and (b) extra permissions not explicitly described by the assertion. Also recall that non-unique inputs in a heap allow arbitrary subsequent behavior through \ContradictRule\  and \ChaosRule\  rules.
We will therefore construct our canonical models of a process $P$ with respect to an input schedule $\rho$, which uniquely determines the inputs read by the process. 
These observations lead us to the construction of a canonical (heap) model $\cmodel(P,\rho)$ for each input schedule $\rho$. The set $\can(P)$ contains such a model for each input schedule $\rho$. The construction of $\cmodel(P,\rho)$ satisfies the following properties:
%
\begin{itemize}
\item $\cmodel(P,\rho) \models \emb(P,t_{\nil})$, \ie, canonical models are indeed models of $\emb(P,t_{\nil})$, where $t_{\nil}$ is the distinguished starting place of $\cmodel(P,\rho)$. 

\item A token never returns to the same place: the I/O permissions of $\cmodel(P,\rho)$ induce a tree on places where each $\bio(t,v,w,t')$ gives rise to an edge from $t$ to $t'$; this solves problem~(a).

\item $\cmodel(P,\rho)$ does not contain any extra permissions, \ie, every proper sub-multiset of $\cmodel(P,\rho)$ fails to satisfy $\emb(P,t_{\nil})$. This addresses problem~(b). We do not explicitly prove this property, but some of our trace inclusion proofs implicitly rely on it.
\end{itemize}
Intuitively, we construct a canonical heap model, $\cmodel( P,\rho)$, given a process $P$ and an input schedule $\rho$ by transforming the (syntactic) tree of $P$ for the input schedule $\rho$ to a corresponding heap model.

An \emph{input schedule} is a function $\rho : \ActBio^* \times \Bio \times \Val \fun \Val$ mapping an I/O trace $\tau$, an I/O operation $\bio$, and an output value $v$ to an input value $\rho(\tau,\bio,v)$. An input schedule $\rho$ is well-typed, written $\welltyped(\rho)$, if $\rho(\tau, \bio, v) \in \Ty(\bio,v)$ for all $\tau$, $\bio$, and~$v$.
We use the set of positions in a binary tree as our set of places $\Place = \List{\{\Left,\Right\}}$. 
%

We then construct the canonical model $\cmodel(P, \rho)$ in two steps (see also Example~\ref{ex:canonical-model}):
\begin{enumerate}
\item We define a recursive function $\premodel$, where
$\premodel(P,\rho,\tau,ppos,cpos,pos)$ returns a singleton multiset containing an I/O permission $\bio(t, v, w, t')$ corresponding to the I/O operation at position $pos$ of process $P$ under the input schedule $\rho$ (if any, otherwise~$\mempty$). Its starting place is given by $t=ppos$, where $ppos$ is the position $ppos$ of the last process appearing directly under an I/O operation prefix (initially $\nil$). Its target place is $t'=cpos \cc \mklist{L}$, where $cpos$ is the current position in the original process (\ie, the path already traversed). The trace $\tau$ records the traversed I/O actions and is used to determine the scheduled input $w = \rho(\tau, \bio, v)$. 
More precisely, for a prefix process $P = \bio(v, z).P'$, $\premodel$ behaves as follows. If $pos=\nil$, it returns the corresponding I/O permission $\bio(t, v, w_{\rho}, t')$, where $w_{\rho} = \rho(\tau, \bio, v)$ is the scheduled input, $t = ppos$ is the starting place and $t' = cpos\cc\mklist{\Left}$ is the target place. If $pos = \Left \cons pos'$ then the prefix is ``traversed'', calling $\premodel$ recursively with process $P[w_{\rho}/z]$ the updated trace $\tau \cc \mklist{\bio(v, w_{\rho})}$ and updated previous position $cpos \cc \mklist{\Left}$. Otherwise, it returns $\mempty$. Choices are traversed recursively.
Figure~\ref{fig:premodel} shows the formal definition of the function $\premodel$, which we discuss below.

\item We define the canonical heap model for a process $P$ and input schedule~$\rho$ by $\cmodel(P, \rho) = \gmodel(P, \rho, \nil, \nil, \nil)$, \ie, as an instance of an auxiliary function $\gmodel(P, \rho, \tau, ppos, cpos)$, 
%
which is defined by collecting all I/O permissions generated by the function $\premodel$ using the multiset sum over all positions $pos$: 
\[
\gmodel(P, \rho, \tau, ppos, cpos) = \bigmultisetsum_{pos} \premodel(P, \rho, \tau, ppos, cpos, pos)
\]

\end{enumerate}

We also define some derived heaps, adding a token to a canonical model, indicated by a superscript, \ie, $\addtoken{\gmodel}(P, \rho, \tau, ppos, cpos) = \mset{\token(ppos)} \mplus \gmodel(P, \rho, \tau, ppos, cpos)$ and $\addtoken{\cmodel}(P, \rho) = \addtoken{\gmodel}(P, \rho, \nil, \nil, \nil)$.
Finally, we define the set $\can(P)$ of canonical heap models of $P$ by 
\[
\can(P) = \{ h \mid \exists \rho.\,  \welltyped(\rho) \land h = \addtoken{\cmodel}(P, \rho)  \}.
\]

\begin{figure}[t]
\begin{center}
\begin{align*}
\premodel(\bio(v, z).P, \rho, \tau, ppos, cpos, \nil) 
& = \mset{\bio(ppos, v, \rho(\tau,\bio,v), cpos \cc \mklist{\Left}) }
\\ 
\premodel(\bio(v, z).P, \rho, \tau, ppos, cpos, \Left \cons pos) 
& = \premodel(P[w_\rho/z]), \rho, \tau \snoc{\bio(v, w_\rho)}, cpos\snoc{\Left}, cpos\snoc{\Left}, pos)
\\ 
\premodel(P_1 \choice \_,\rho,\tau,ppos,cpos,\Left \cons pos) 
& = \premodel(P_1, \rho, \tau, ppos, cpos\snoc{\Left}, pos)
\\ 
\premodel(\_ \choice P_2,\rho,\tau,ppos,cpos,\Right \cons pos) 
& = \premodel(P_2, \rho, \tau, ppos, cpos\snoc{\Right}, pos)
\\
\premodel(\_, \_, \_, \_, \_, \_)  
& = \mempty, \mathrm{otherwise.}
\end{align*}
\caption{Function $\premodel$ maps a process and a position to a singleton multiset containing an I/O permission or $\mempty$.
We abbreviate $\rho(\tau,\bio,v)$ as $w_\rho$ and write $x\cons xs$ for prefixing an element to a list (cons).
}
\label{fig:premodel}
\end{center}
\end{figure}

Figure~\ref{fig:premodel} shows the formal definition of the function $\premodel$.
The first two equations defining $\premodel$ cover the case of a prefixed process $\bio(v,w).P$. In the first equation, the desired position is reached ($pos=\nil$) and a singleton multiset containing the  corresponding I/O permission with source place $ppos$, target place $cpos\snoc{L}$, and scheduled input $\rho(\tau,bio,v)$ is returned. In the second equation, the position has head $\Left$ and the search continues in the process $P[w_\rho/z]$, where the scheduled input $w_\rho=\rho(\tau,\bio,v)$ replaces the bound variable $z$, for the trace $\tau$ extended with the traversed I/O event $\bio(v, w_\rho)$ and with the arguments $ppos$ and $cpos$ both set to $cpos \snoc{L}$. 
The third and fourth equations recursively navigate into a choice process in the direction given by the position, updating $cpos$ but not $ppos$ in the recursive call. The final equation catches all cases not covered by the previous equations and returns the empty multiset. 
Note that the concatenation $cpos \cdot pos$ is invariant throughout the recursive calls.

\subsubsection{Canonical Model Property}

The following lemma provides fixed-point equations for the canonical models, with one case per process form:
\begin{lemma}[Canonical model as fixed-point] \mbox{ }
\label{lem:gmodel-charact}
\begin{enumerate}
\item $\gmodel(\Null, \rho, \tau, ppos, cpos) = \mempty$,
\item $\gmodel(\bio(v, z).P, \rho, \tau, ppos, cpos) = $ \\
         $\mset{\bio(ppos, v, w_{\rho}, cpos')} \mplus 
          \gmodel(P[w_{\rho}/z], \rho, \tau \snoc{\bio(v, w_{\rho})}, cpos', cpos')$ \\
          where $w_{\rho} = \rho(\tau, \bio, v)$ and $cpos' = cpos\snoc{\Left}$, and
\item $\gmodel(P_1 \choice P_2, \rho, \tau, ppos, cpos) = $ \\
          $ \gmodel(P_1, \rho, \tau, ppos, cpos \snoc{\Left}) \mplus 
           \gmodel(P_2, \rho, \tau, ppos, cpos \snoc{\Right})$.
\end{enumerate}
\end{lemma}

\begin{proposition}[Canonical model property]
${\cmodel}(P, \rho) \models \emb(P,\nil)$ for all processes $P$ and well-typed schedules $\rho$.
\end{proposition}
\begin{proof}
The lemma's statement follows from $\gmodel(P, \rho, \tau, ppos, cpos) \models \emb(P, ppos)$, which we prove by coinduction using the relation $X$ on heaps and formulas defined by 
\[
X(h,\phi) = \exists P,\tau,ppos,cpos.\, h = \gmodel(P,\rho,\tau,ppos,cpos) \land \phi = \emb(P, ppos),
\]
and a case analysis on the structure of $P$. The different cases are proved using the fixed point property of $\gmodel$ stated in Lemma~\ref{lem:gmodel-charact}.
\end{proof}

\subsection{Formal Definitions and Proofs for Section~\ref{ssec:processes-simulate-canonical-models}}
\label{sapp:processes-simulate-canonical-models}

We now turn to the second trace inclusions of Equation~\eqref{eq:trace-inclusions}, given on page~\pageref{eq:trace-inclusions}. It states that each trace of the canonical model set $\can(P)$ is also a trace of $P$.
We would like heap transitions of the canonical model $\addtoken{\cmodel}(P,\rho)$ to lead to a heap $\addtoken{\cmodel}(P',\rho)$ for some process $P'$, so that we can simulate it with the corresponding process transition from $P$ to $P'$. Recall that there are two problems:
\begin{enumerate}
\item \label{prob:dead-heap} {Dead heap parts:} 
Consider the process $P = Q \choice \bio(v, z).R$. The canonical model $\addtoken{\cmodel}(P,\rho)$ has a transition labeled $\bio(v,w)$ to $\addtoken{\gmodel}(R[w/z], \rho, \mklist{\bio(v,w)}, cpos', cpos') \mplus g$ where $w=\rho(\nil,\bio,v)$, $g = \addtoken{\cmodel}(Q, \rho)$, and $cpos' = \mklist{\Right,\Left}$ with the resulting token at place $cpos'$ (see Lemma~\ref{lem:gmodel-charact}). Since this token can subsequently only visit places in $\{pos \mid cpos' \leq pos \}$, this means that the $g$ portion of the heap will never be able to make a transition. Our proof must take such \emph{dead} heap parts into account.

\item \label{prob:chaotic-trans} {Chaotic transitions:} Let $\tau$ be a trace of all canonical models $\addtoken{\cmodel}(P,\rho)$ (\ie, for all input schedules~$\rho$). Some of these models transit to the ``chaotic'' state $\bot$ at some point along the trace. However, a given process cannot in general simulate the (arbitrary) I/O actions possible in that state. We will have to carefully pick a particular schedule to avoid this problem.
\end{enumerate}

We address problem~\eqref{prob:dead-heap} by considering heaps with dead parts in our transition lemmas. Let $\srcplaces(h)$ be the set of source places occurring in I/O permissions in the heap $h$. A heap $h$ is called \emph{dead} with respect to a position $pos$ if $h$ contains no tokens and $\srcplaces(h) \cap \{pos' \mid pos \leq pos'\} = \emptyset$, meaning a token at position $pos$ in a canonical model will never activate a transition in $h$.

The solution to problem~\eqref{prob:chaotic-trans} is based on the observation that executing some I/O action $\bio(v, w_{\rho})$ with \emph{scheduled input $w_{\rho}=\rho(\tau, \bio, v)$} from $\addtoken{\cmodel}(P,\rho)$ indeed leads to a heap $\addtoken{\cmodel}(P',\rho)$ for some process $P'$ (and, in particular, not to~$\bot$).
Hence, to simulate a given trace $\tau$ of the heap $\addtoken{\cmodel}(P, \rho)$ by transitions of the process $P$, we must ensure that the schedule $\rho$ is consistent with the trace $\tau$. We therefore define a ``witness'' schedule $\rhowit(\tau)$ such that, roughly speaking, 
\[
\addtoken{\cmodel}(P,\rhowit(\tau))\trans{\tau} h = \addtoken{\cmodel}(P',\rhowit(\tau))  
\]
for some process $P'$, \ie, in particular, $h \neq \bot$. We define the schedule $\rhowit(\tau)$
 to return the inputs appearing on the trace $\tau$:
\begin{align*}
\rhowit(\bio'(v',w) \cons \tau, (\nil, \bio, v)) & = \ifte{\bio' = \bio \land v' = v}{w}{\pick(\Ty(\bio,v))} \\
\rhowit(a \cons \tau', (b \cons \tau, bio, v)) & = \ifte{a = b}{\rhowit(\tau', (\tau, bio, v))}{\pick(\Ty(\bio,v))} \\
\rhowit(\_, (\_, bio, v)) & = \pick(\Ty(\bio,v)).
\end{align*}
That is, for proper prefixes $\tau'$ of the trace $\tau$, I/O operation $\bio$, and output $v$, the schedule $\rhowit(\tau)$ returns the input $w$, if $\bio(v,w)$ is the next step in $\tau$ after the prefix $\tau'$. For other traces, it returns an arbitrary well-typed input (\ie, $\pick(S)$ selects an arbitrary element from a non-empty set $S$). 


The following three lemmas make the intuition given above more precise.
We first prove a lemma about the individual transitions of canonical models.
%
\begin{lemma}[Canonical heap transitions]
\label{lem:cmod-transition}
Suppose that
\[
\addtoken{\gmodel}(P, \rho, \tau, ppos, cpos) \mplus g \trans{\bio(v, w)} h
\]
for some heaps $g$ and $h$ such that $w\in\Ty(bio,v)$, $ppos \leq cpos$, and $g$ is dead for $ppos$. Then there exist a process $P'$, positions $cpos'$ and $pos'$, and a heap~$g'$ such that $w = \rho(\tau, \bio, v)$, $g'$ is dead for $cpos'$,   $P \trans{\bio(v, w)} P'$,
 and
$h = \addtoken{\gmodel}(P', \rho, \tau \snoc{\bio(v, w)}, cpos', cpos') \mplus g'.$
\end{lemma}

The following lemma
states that transitions with scheduled input never lead to the chaotic state~$\bot$.
%
\begin{lemma}[Transitions with scheduled input]
\label{lem:scheduled-input-trans}
If $\addtoken{\gmodel}(P, \rho, \tau, ppos, cpos) \trans{\bio(v,w)} ho$ for $ppos \leq cpos$ and $w=\rho(\tau, \bio, v)$, then $ho \neq \bot$.
\end{lemma}
%

Next, we extend these lemmas from individual transitions to traces. 


\begin{lemma}[Canonical heap traces]
\label{lem:opsem-simulates-cmod-trace}
Suppose we have 
\[
\addtoken{\gmodel}(P, \rhowit(\sigma), \tau, ppos, cpos) \mplus g \trans{\tau'} ho,
\]
where $\tau \cc \tau' \leq \sigma$, $ppos \leq cpos$, $g$ is dead for $ppos$, and $ho \in\option{\Heap}$. Then there exist a process $P'$, place~$t'$, heap~$g'$, and positions $ppos'$ and $cpos'$ such that $ppos' \leq cpos'$, $g'$ is dead for $ppos'$,
\[
P \trans{\tau'} P', \text{ and }
ho = \addtoken{\gmodel}(P', \rhowit(\sigma), \tau \cc \tau', ppos', cpos') \mplus g'.
\]
\end{lemma}

\begin{proof}
By trace induction using Lemmas~\ref{lem:cmod-transition} and~\ref{lem:scheduled-input-trans} for single transitions.
\end{proof}

Now we can prove that each trace of the set of canonical models of $P$ is also a trace of $P$.
\begin{proposition}
$\tracesH(\can(P)) \subseteq \traces(P)$.
\end{proposition}

\begin{proof}
Let $\tau \in \tracesH(\can(P))$. Then $\tau$ is a trace of all canonical heap models of $P$; hence, in particular, $\addtoken{\cmodel}(P,\rhowit(\tau)) \trans{\tau} ho$ for some $ho \in \Heap_\bot$. By Lemma~\ref{lem:opsem-simulates-cmod-trace}, we conclude that $\tau\in \traces(P)$.
\end{proof}

\else 
\fi

\end{document}